\documentclass[pdflatex,sn-mathphys-num]{sn-jnl}

\usepackage{graphicx}%
\usepackage{multirow}%
\usepackage{amsmath,amssymb,amsfonts}%
\usepackage{amsthm}%
\usepackage{mathrsfs}%
\usepackage[title]{appendix}%
\usepackage{xcolor}%
\usepackage{textcomp}%
\usepackage{manyfoot}%
\usepackage{booktabs}%
\usepackage{algorithm}%
\usepackage{algorithmicx}%
\usepackage{algpseudocode}%
\usepackage{listings}%
\usepackage{hyperref}
\usepackage{subfig}
\usepackage{xspace}
\usepackage{bm}
\usepackage{graphicx}
\usepackage{multirow}
\usepackage{booktabs}
\usepackage{hyperref}

\newcommand{\mX}{\mathbf{X}}
\newcommand{\mY}{\mathbf{Y}}
\newcommand{\mF}{\mathbf{F}}
\newcommand{\mG}{\mathbf{G}}
\newcommand{\mP}{\mathbf{P}}
\newcommand{\mK}{\mathbf{K}}

\newcommand{\vx}{\mathbf{x}}
\newcommand{\vy}{\mathbf{y}}
\newcommand{\vt}{\mathbf{t}}

\makeatletter
\DeclareRobustCommand\onedot{\futurelet\@let@token\@onedot}
\def\@onedot{\ifx\@let@token.\else.\null\fi\xspace}

\def\eg{\emph{e.g}\onedot}

\def\ie{\emph{i.e}\onedot}

\def\vs{\emph{vs}\onedot}

\makeatother

\theoremstyle{thmstyleone}%
\newtheorem{theorem}{Theorem}
\newtheorem{proposition}[theorem]{Proposition}%

\theoremstyle{thmstyletwo}%

\theoremstyle{thmstylethree}%

\begin{document}

\title[Temporally-Aligned Evaluation for Audio-Driven Talking Head Generation]{Temporally-Aligned Evaluation for Audio-Driven Talking Head Generation}







\author[1]{\fnm{Zhicheng} \sur{Zhang}}\email{zhicheng.zhang2@unsw.edu.au}
\equalcont{Zhicheng Zhang and Lei Wang contributed equally to this work and are co-first authors.}

\author[2,3]{\fnm{Lei} \sur{Wang}}\email{l.wang4@griffith.edu.au}
\equalcont{Zhicheng Zhang and Lei Wang contributed equally to this work and are co-first authors.}

\author*[1]{\fnm{Yu} \sur{Zhang}}\email{m.yuzhang@unsw.edu.au}

\author*[2]{\fnm{Yongsheng} \sur{Gao}}\email{yongsheng.gao@griffith.edu.au}

\affil[1]{\orgname{School of Business, University of New South Wales (UNSW)}, 
\orgaddress{\country{Australia}}}

\affil[2]{\orgname{School of Engineering and Built Environment, Griffith University}, 
\orgaddress{\country{Australia}}}

\affil[3]{\orgname{Data61/CSIRO}, 
\orgaddress{\country{Australia}}}



\abstract{Audio-driven talking-head generation has advanced rapidly, yet existing evaluation protocols mainly rely on frame-wise metrics that assume strict temporal correspondence between generated and reference videos. This assumption does not match speech-driven facial motion, which naturally includes slight timing shifts, different speaking speeds, and stylistic variations. As a result, conventional metrics may treat harmless timing differences as quality errors, making it harder to fairly compare methods and understand their trade-offs. 
In this work, we argue that evaluation of dynamic generative models should be formulated as a sequence-alignment problem rather than independent frame comparison. We introduce a unified sequence-level reformulation that integrates Soft Dynamic Time Warping into established evaluation pipelines. By aligning feature trajectories while preserving temporal order, the proposed framework provides robustness to bounded temporal misalignments without altering the underlying perceptual, identity, or synchronization encoders. We show that frame-wise evaluation can be viewed as a special case under rigid alignment, while sequence-level alignment provides improved stability, lower sensitivity to timing differences, and clearer separation between modeling paradigms. Building on this principled formulation, we conduct a large-scale benchmark of 20 methods across seven datasets spanning canonical, in-the-wild, and style-diverse scenarios under standardized protocols. Extensive experiments show that temporally aligned metrics are more robust to timing differences, provide more consistent results across datasets, and better reveal systematic trade-offs between modeling paradigms, such as synchronization versus realism and expressiveness versus stability. By redefining talking-head evaluation as trajectory-level alignment, this work highlights temporal correspondence as an essential principle for evaluating dynamic generative models and provides a solid foundation for future research in audio-driven facial animation.}




\maketitle

\section{Introduction}\label{sec1}

Audio-driven talking-head generation has progressed rapidly in recent years, fueled by advances in generative modeling \cite{shen2023difftalk,wei2024aniportrait,zhang2023sadtalker,croitoru2023diffusion,xing2024survey}, motion representation learning \cite{li2025instag,fan2022faceformer,wang2021audio2head}, and audio-visual alignment \cite{prajwal2020lip,zhou2019talking,chung2016out}. Contemporary approaches can synthesize photorealistic facial animations from speech with high lip-sync accuracy and expressive dynamics, enabling applications in virtual avatars\cite{fan2022faceformer,wang2021audio2head,thies2020neural}, digital dubbing\cite{zhou2020makelttalk,har2026heygen}, telepresence\cite{jin2020live,christoff2023application}, and interactive media\cite{zhu2025infp,yan2024dialoguenerf}.

Despite these advances, progress in the field remains difficult to assess in a principled and consistent manner\cite{bai2025survey,zhen2023human,gowda2023pixels}. Existing studies typically evaluate models on isolated datasets, adopt heterogeneous preprocessing pipelines, and report results under inconsistent metric configurations\cite{zhang2021flow,chung2018voxceleb2,wang2020mead}. More fundamentally, widely used evaluation metrics, such as Learned Perceptual Image Patch Similarity (LPIPS)\cite{Zhang2018LPIPS,snell2017learning} for perceptual similarity, Cosine Similarity (CSIM)\cite{deng2019arcface} for identity preservation, SyncNet-based synchronization scores\cite{Chung2016SyncNet}, and interpolation-based motion smoothness\cite{huang2024vbench}, share a common structural limitation: they operate under rigid frame-wise alignment assumptions\cite{petitjean2011global}.

\begin{figure}[tbp]
\centering
\includegraphics[trim=0 0 0 0, clip=true, width=\linewidth]{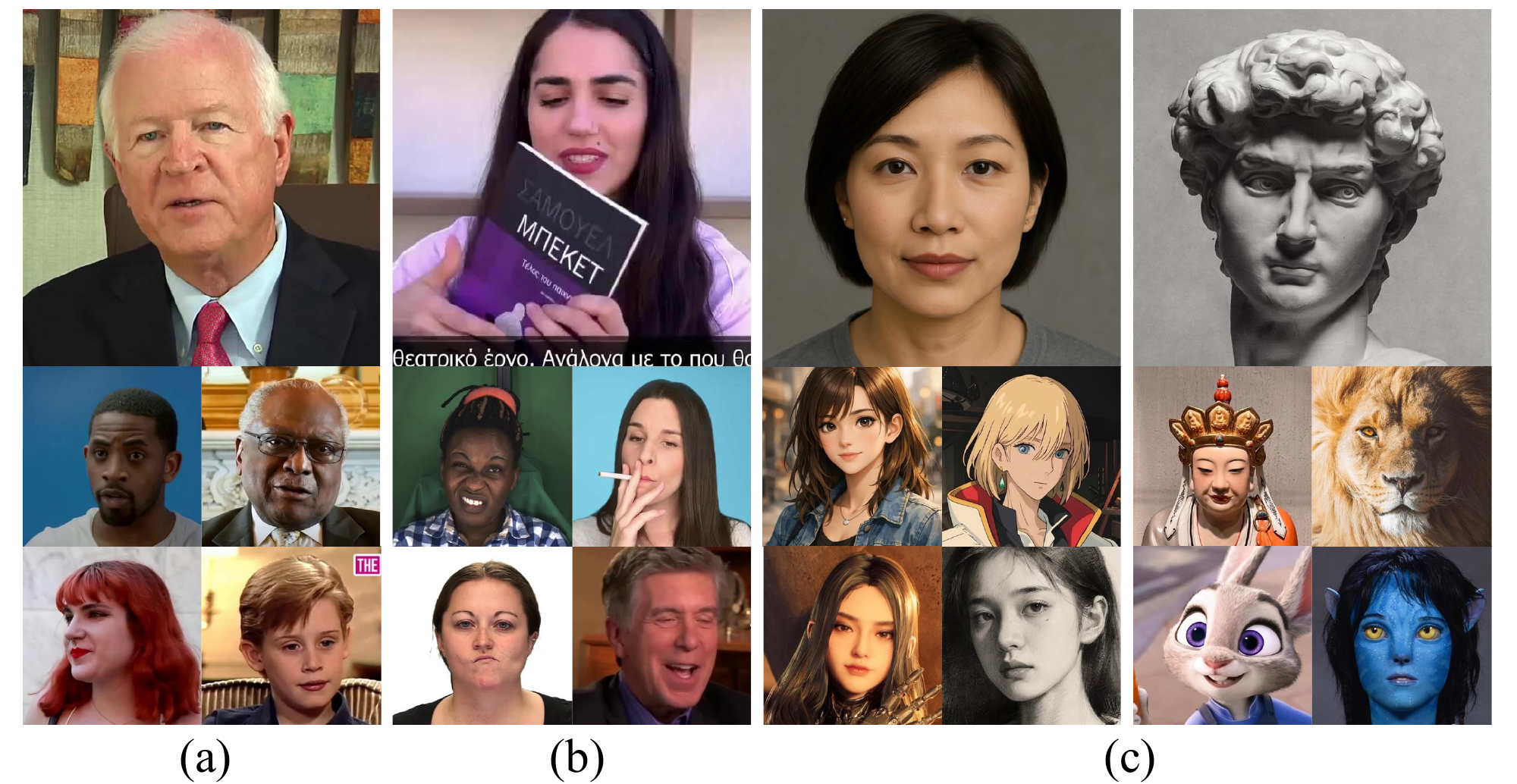}
\caption{Comparison of datasets across columns: existing benchmarks with diverse real identities (1st); our \textit{Wild} dataset with occlusions, exaggerated motions, and dynamic expressions (2nd); our \textit{Avatar} dataset including AI-generated photorealistic and 3D faces, animated characters, 2D cartoons, sculptures, artistic portraits, sketch-style renderings, animals, and humanoids (3rd-4th columns).
}
\label{fig:dataset}
\end{figure}

Specifically, these metrics compute similarity between temporally corresponding frame pairs and aggregate the results by simple averaging. This formulation implicitly assumes that generated and reference videos are strictly synchronized, have equal length, maintain stable identity, and exhibit consistent audio-visual correspondence. However, in practical audio-driven generation, even high-quality models may introduce slight temporal phase shifts, rhythm variations, pose changes, or local articulation delays. While such variations are often perceptually acceptable, they can disproportionately distort frame-level evaluation scores. As a result, current metrics may reflect sensitivity to temporal misalignment rather than intrinsic generation quality\cite{cuturi2017soft}.

This structural issue becomes increasingly problematic as the field diversifies into multiple modeling paradigms, including lip-centric pipelines\cite{ma2023dreamtalk,zhang2024musetalk,wang2023seeing,prajwal2020lip}, motion-space disentangled frameworks\cite{li2025ditto,liu2024anitalker,cao2024joyvasa,tan2024edtalk,wang2021audio2head,zhang2023sadtalker}, multi-condition fusion strategies\cite{wang2024v,chen2025echomimic,zhong2023identity}, and holistic full-motion architectures\cite{ki2025float,xu2024hallo,cui2024hallo2,cui2025hallo3,zheng2024memo,wei2024aniportrait,ji2025sonic}. Without temporally robust evaluation, it remains unclear which paradigms truly excel in realism, synchronization, identity stability, or motion naturalness, and which trade-offs are intrinsic to their design.

Although several metrics incorporate limited temporal modeling, they do not address the structural misalignment problem inherent in frame-constrained evaluation. For example, Fr\'echet Video Distance (FVD) evaluates distributional similarity using spatiotemporal features extracted from pretrained video recognition networks. While FVD captures global temporal statistics, it operates at the distribution level rather than establishing explicit alignment between generated and reference sequences. Consequently, it cannot distinguish perceptually acceptable phase shifts from genuine motion inconsistencies.

Similarly, some works adopt sliding-window averaging, short temporal smoothing, or local aggregation strategies to reduce frame-level noise. These approaches mitigate high-frequency fluctuations but still assume fixed temporal correspondence and cannot accommodate non-linear timing variations. Cross-correlation-based synchronization metrics partially account for temporal offsets by measuring alignment peaks across short windows; however, they lack monotonic alignment constraints and are limited to specific modalities such as lip synchronization.
In contrast, audio-driven talking-head generation involves structured, monotonic temporal evolution governed by speech dynamics. Evaluation therefore requires a mechanism that preserves temporal order while allowing elastic alignment under perceptually tolerable timing variations. This observation motivates our reformulation of talking-head evaluation as sequence-level trajectory alignment, rather than distribution matching or local smoothing.

\begin{table}[tbp]
\setlength{\tabcolsep}{3pt}
\renewcommand{\arraystretch}{0.7}
\centering
\caption{Overview of representative audio-driven talking-head generation methods and their public implementations.
For each method, we report the publication venue and year, the availability of official GitHub code, and the number of GitHub stars (collected Jun 01, 2026), providing a consolidated view of methodological coverage and practical adoption. The hierarchical semantic taxonomy in the last column was built via unsupervised hierarchical clustering of paper titles and abstracts using TF-IDF representations and Ward linkage. 
The horizontal axis denotes linkage distance, reflecting semantic dissimilarity between methods. This taxonomy offers a structured view of how existing approaches differ in motion modeling, conditioning mechanisms, and system-level design.}
\label{tab:talkinghead_methods}

\scriptsize
\begin{tabular}{l c c r c}
\toprule
Method & Venue & GitHub & Stars & Hierarchical semantic taxonomy \\
\midrule

\multicolumn{4}{l}{\textbf{2025}} &
\multirow{23}{*}{%
\includegraphics[angle=90,width=0.38\linewidth]{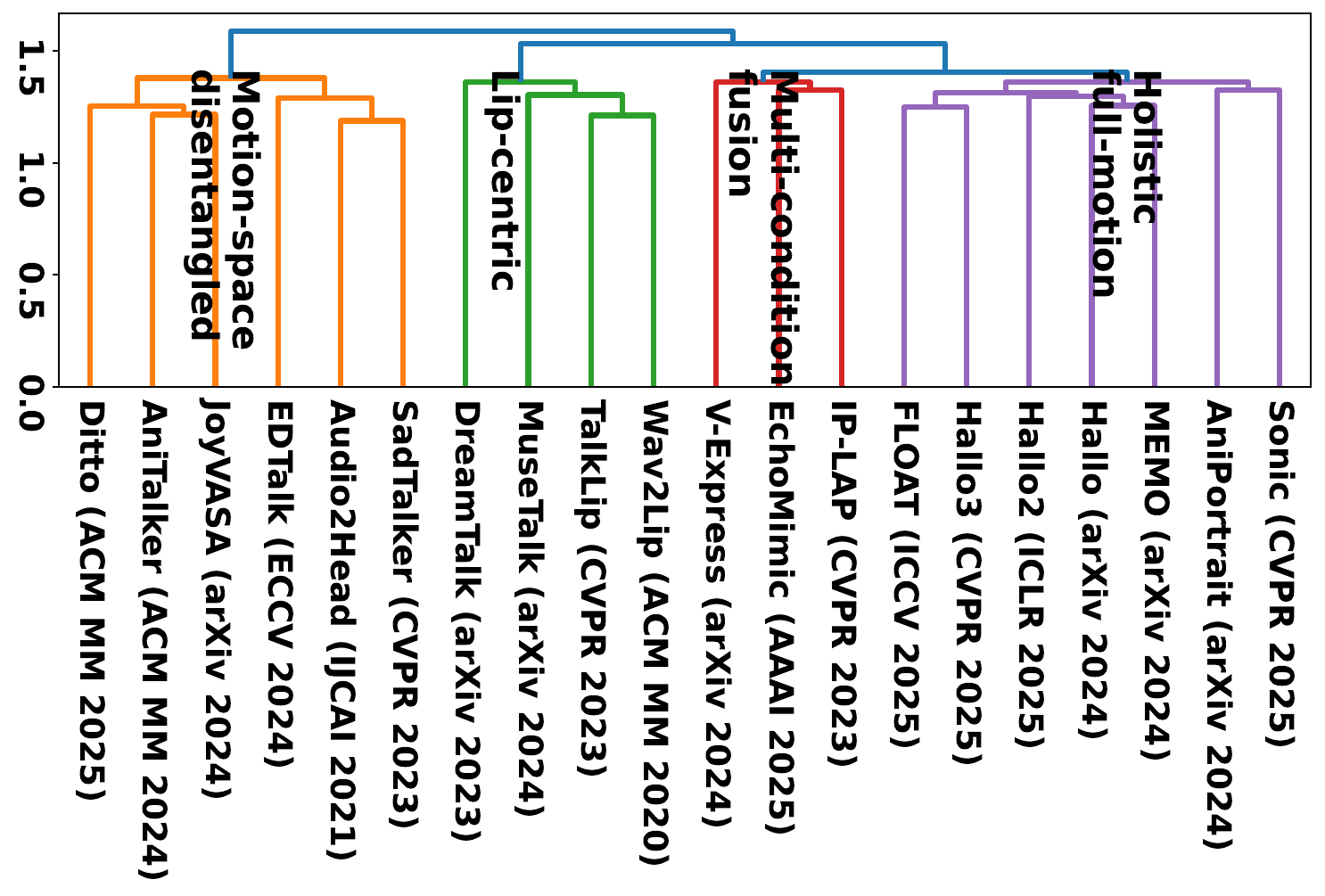}%
}
\\
Ditto~\cite{li2025ditto}                & ACM MM & \href{https://github.com/antgroup/ditto-talkinghead}{Code} & 805  &  \\
EchoMimic~\cite{chen2025echomimic}      & AAAI   & \href{https://github.com/antgroup/echomimic}{Code}        & 4200 & \\
FLOAT~\cite{ki2025float}                & ICCV   & \href{https://github.com/deepbrainai-research/float}{Code} & 474  & \\
Hallo3~\cite{cui2025hallo3}             & CVPR   & \href{https://github.com/fudan-generative-vision/hallo3}{Code} & 1440 & \\
Sonic~\cite{ji2025sonic}                & CVPR   & \href{https://github.com/jixiaozhong/Sonic}{Code}          & 3200 & \\
Hallo2~\cite{cui2024hallo2}             & ICLR   & \href{https://github.com/fudan-generative-vision/hallo2}{Code} & 3741 &  \\

\multicolumn{4}{l}{\textbf{2024}} & \\
Hallo~\cite{xu2024hallo}                & arXiv  & \href{https://github.com/fudan-generative-vision/hallo}{Code} & 8600 & \\
AniPortrait~\cite{wei2024aniportrait}   & arXiv  & \href{https://github.com/scutzzj/aniportrait}{Code}        & 5018 & \\
MuseTalk~\cite{zhang2024musetalk}       & arXiv  & \href{https://github.com/tmelyralab/musetalk}{Code}       & 5900 &  \\
V-Express~\cite{wang2024v}              & arXiv  & \href{https://github.com/tencent-ailab/V-Express/}{Code}  & 2430 &  \\
AniTalker~\cite{liu2024anitalker}       & ACM MM & \href{https://github.com/x-lance/anitalker}{Code}         & 1600 & \\
MEMO~\cite{zheng2024memo}               & arXiv  & \href{https://github.com/memoavatar/memo.git}{Code}       & 1118 & \\
JoyVASA~\cite{cao2024joyvasa}           & arXiv  & \href{https://github.com/jdh-algo/JoyVASA}{Code}          & 870  & \\
EDTalk~\cite{tan2024edtalk}             & ECCV   & \href{https://github.com/tanshuai0219/EDTalk?tab=readme-ov-file}{Code} & 469 & \\

\multicolumn{4}{l}{\textbf{2023}} & \\
SadTalker~\cite{zhang2023sadtalker}     & CVPR   & \href{https://github.com/winfredy/sadtalker}{Code}        & 13905 & \\
DreamTalk~\cite{ma2023dreamtalk}        & arXiv  & \href{https://github.com/ali-vilab/dreamtalk}{Code}       & 1834  & \\
IP-LAP~\cite{zhong2023identity}         & CVPR   & \href{https://github.com/Weizhi-Zhong/IP_LAP}{Code}       & 740   & \\
TalkLip~\cite{wang2023seeing}           & CVPR   & \href{https://github.com/sxjdwang/talklip}{Code}          & 428   & \\

\multicolumn{4}{l}{\textbf{2022--2020}} & \\
Audio2Head~\cite{wang2021audio2head}    & IJCAI  & \href{https://github.com/wangsuzhen/Audio2Head}{Code}     & 353   & \\
Wav2Lip~\cite{prajwal2020lip}           & ACM MM & \href{https://github.com/Rudrabha/Wav2Lip}{Code}          & 13000 & \\

\bottomrule
\end{tabular}
\end{table}

In this work, we identify rigid frame-wise alignment as a fundamental bottleneck in talking-head evaluation and propose a unified temporal reformulation of existing metrics. Instead of treating evaluation as independent frame comparisons, we reinterpret it as sequence-level trajectory matching in feature space. Concretely, we introduce a Soft Dynamic Time Warping (Soft-DTW) based alignment mechanism that preserves temporal order while allowing flexible correspondence between generated and reference sequences.
Importantly, this reformulation does not alter feature encoders. 
The modification occurs solely at the metric level, transforming frame-wise averaging into temporally aligned sequence matching. In addition, we redesign motion smoothness evaluation by shifting from pixel-space interpolation errors to explicit semantic motion trajectory modeling using disentangled head pose and expression features, further enhancing interpretability and robustness.

Built upon this unified temporal evaluation framework, we conduct a large-scale benchmark of 20 high-impact pretrained methods across seven datasets under standardized protocols. Our analysis reveals consistent paradigm-level trade-offs and demonstrates that temporally aware metrics provide more stable, discriminative, and practically meaningful comparisons.
By establishing temporal alignment as a first-class principle in evaluation, this work provides a fair and robust foundation for future research in audio-driven talking-head generation.

Our contributions are threefold:
\renewcommand{\labelenumi}{\roman{enumi}.}
\begin{enumerate}
\item A unified temporal reformulation of evaluation metrics. We identify rigid frame-wise alignment as a structural limitation in existing talking-head evaluation protocols and reformulate perceptual similarity (LPIPS), identity preservation (CSIM), and audio-visual synchronization (SyncNet) as sequence-level trajectory matching problems using Soft-DTW. This modification introduces temporal robustness without altering feature encoders.
\item Semantic motion trajectory modeling for expression naturalness. We redesign motion smoothness evaluation by replacing interpolation-based pixel reconstruction errors with explicit semantic motion features extracted from a motion encoder, disentangling rigid head pose and non-rigid expression deformation. Combined with Soft-DTW alignment, this formulation provides interpretable and temporally robust motion evaluation.
\item A standardized large-scale benchmark with paradigm-level insights. We evaluate 20 high-impact audio-driven talking-head models across seven canonical and real-world datasets under unified protocols. Our temporally aware metrics reveal consistent paradigm-specific trade-offs across realism, synchronization, identity stability, and motion naturalness, offering actionable guidance for model design and deployment.
\end{enumerate}

\section{Related Work}\label{sec2}
\textbf{Surveys and reviews.} Several studies \cite{zhang2026talkinghead} summarize the landscape of audio-driven talking-head generation, reviewing architectures, motion representations, and disentanglement strategies. While useful for method selection, these surveys are largely qualitative and descriptive, lacking quantitative, cross-dataset evaluation or temporal alignment considerations. As a result, they cannot reveal systematic trade-offs between paradigms under realistic generation conditions.

\textbf{Existing evaluations and benchmarks.} A few works \cite{shen2023difftalk,zhang2023sadtalker} perform benchmarking using metrics such as LPIPS, Fr\'echet Inception Distance (FID), PSNR, CSIM, or SyncNet-based lip-sync scores on one or two datasets. Although informative, these studies are limited in scope, covering a small subset of methods, datasets, and evaluation dimensions. Most metrics assume strict frame-wise alignment, ignoring minor phase shifts, rhythm variations, or temporal dynamics that are perceptually acceptable yet critical for faithful video assessment. Moreover, prior evaluations rarely integrate multiple semantic dimensions (\eg, perceptual quality, identity, synchronization, motion smoothness) into a single, interpretable framework, making it difficult to compare paradigms, quantify trade-offs, or assess robustness across datasets and styles.

\textbf{Datasets.} Evaluation quality critically depends on dataset diversity and preprocessing. Widely used 2D talking-head datasets, including HDTF \cite{zhang2021flow}, VoxCeleb2 \cite{chung2018voxceleb2}, CelebV-HQ \cite{zhu2022celebv}, MEAD \cite{wang2020mead}, and RAVDESS \cite{livingstone2018ryerson}, are often employed selectively, with heterogeneous preprocessing, limited subsets, and insufficient style, pose, or lighting variations. Smaller domain-shift datasets \cite{cui2025hallo3,li2025ditto,chen2025echomimic} mainly support qualitative validation. To address these limitations, we construct a comprehensive benchmark spanning seven datasets, including two curated subsets (see Fig. \ref{fig:dataset}): (i) \textit{Wild}, emphasizing challenging real-world conditions, occlusions, expressive speech, and diverse head motion; and (ii) \textit{Avatar}, covering AI-generated photorealistic, 2D/3D animated, artistic, and non-human faces. This strategy enables cross-dataset and cross-style evaluation, revealing robustness and generalization trends beyond conventional laboratory conditions.

\textbf{Evaluation metrics.} Previous metrics are fragmented and mostly frame-constrained. Perceptual measures (LPIPS, IQA, FID, FVD, VQA, CPBD) assess visual realism, while pixel-level metrics (PSNR, SSIM, MS-SSIM, L1) quantify reconstruction fidelity. Identity preservation (CSIM), audio-visual synchronization (Sync-C/D), motion naturalness (Smooth), and computational efficiency (FPS) are evaluated independently, often with limited temporal awareness. Our framework unifies all metrics under a Soft-DTW-based, sequence-level alignment approach, accounting for minor temporal misalignments, sequence length variations, and rhythm differences. This produces temporally-aware, interpretable, and comparable scores across all semantic dimensions.

\section{Benchmark Design and Construction}\label{sec3}

We construct a benchmark for audio-driven talking-head generation to enable fair and reproducible comparison of methods across diverse datasets and scenarios.

\subsection{Method Selection and Paradigm Taxonomy}

\textbf{Method selection criteria.}
We conduct a systematic survey of audio-driven 2D talking-head generation methods published between 2020 and 2025, focusing specifically on the audio-to-video setting. An initial pool of 117 candidate methods \cite{zhang2026talkinghead} was collected from arXiv and major peer-reviewed venues. 
To ensure reproducibility and fair cross-method comparison, we apply explicit filtering criteria along three dimensions: publication quality, public accessibility of source code, and availability of pretrained checkpoints. Restricting inclusion to methods with released implementations enables consistent re-evaluation under a unified preprocessing, inference, and metric pipeline, minimizing discrepancies arising from undocumented reimplementations. This criterion prioritizes transparency and experimental verifiability over popularity or citation count, ensuring that all evaluated models can be independently reproduced under identical conditions.
After filtering, 20 representative methods remain, summarized in Table~\ref{tab:talkinghead_methods}. These methods span diverse architectural paradigms and modeling strategies, forming a reproducible and standardized foundation for our paradigm-level benchmark analysis.

\textbf{Motion-modeling paradigm taxonomy. 
} 
We organize audio-driven talking-head methods according to their motion modeling philosophy, reflecting the core assumptions each approach makes about representing and generating facial dynamics. While architectural choices, training protocols, and auxiliary objectives vary widely, these underlying principles provide a more meaningful axis for comparison than surface-level components. Specifically, motion-space disentangled methods explicitly decompose facial motion into interpretable components, such as 3D coefficients, latent motion bases, or dense motion fields \cite{li2025ditto,liu2024anitalker,wang2021audio2head,zhang2023sadtalker}, enhancing controllability, identity preservation, and cross-identity transfer. Lip-centric methods prioritize accurate audio-lip synchronization by modeling mouth dynamics as the primary driver \cite{ma2023dreamtalk,zhang2024musetalk,wang2023seeing,prajwal2020lip}, often achieving high lip accuracy at the expense of full-face realism or temporal coherence. Multi-condition fusion methods integrate audio with additional signals such as identity embeddings, landmarks, emotions, or pose parameters \cite{wang2024v,chen2025echomimic,zhong2023identity}, improving expressiveness and generalization, though cross-dataset robustness and temporal consistency are less frequently evaluated. Finally, holistic full-motion methods map audio (and optional conditioning inputs) directly to complete spatiotemporal facial motion without explicit structural factorization \cite{ki2025float,xu2024hallo,cui2024hallo2,zheng2024memo,wei2024aniportrait}, producing coherent expressions and smooth dynamics, but may be more susceptible to subtle synchronization errors or identity drift. 

\textbf{Semantic clustering validation.
} 
To examine whether our taxonomy reflects intrinsic structure within the literature rather than subjective grouping, we perform unsupervised semantic analysis (see Table \ref{tab:talkinghead_methods}, last column). Specifically, TF-IDF representations are constructed from the title and abstract of each selected paper, followed by hierarchical clustering using Ward linkage \cite{hondru2025masked}. This procedure groups methods based solely on textual semantic similarity without manual supervision.
The resulting dendrogram reveals two high-level semantic divisions: one centered on motion representation and disentanglement, and the other on generation pipelines and system-level modeling. These divisions further decompose into four primary clusters that align closely with the four paradigms identified above. The correspondence between manual paradigm identification and unsupervised semantic clustering provides empirical support for the proposed taxonomy, reducing subjective bias and reinforcing its structural validity.


\subsection{Dataset Selection and Construction}

At the dataset level, we adopt two complementary design principles: literature relevance and robustness coverage. 
We identify the five most frequently used public benchmarks for 2D audio-driven talking-head generation, HDTF, VoxCeleb2, CelebV-HQ, MEAD, and RAVDESS, and construct standardized evaluation subsets under unified preprocessing protocols. This ensures compatibility with prior work while enabling controlled and reproducible cross-method comparison.

For HDTF, VoxCeleb2, and CelebV-HQ, we select 100 videos per dataset, with VoxCeleb2 restricted to high-quality interview recordings to reduce uncontrolled compression artifacts. For RAVDESS, both speaking and singing scenarios are included, selecting at least two videos per identity to yield 100 videos. For MEAD, 47 identities are sampled with a minimum of two videos per identity, again resulting in 100 videos. All videos are trimmed to 4-10 seconds to emphasize dynamic facial motion while preserving identity diversity. For each sample, the initial frame and corresponding audio track are extracted, and audio is resampled to 16 kHz to ensure consistent input conditions.
Following common practices in prior works~\cite{wang2023seeing,chung2016out,zhou2021pose}, we ensure that evaluation results are comparable across methods while focusing on dynamic aspects critical to talking-head generation. The selected subsets balance coverage of motion variability and computational feasibility, without relying on the full test set, which may include excessively long, static, or low-quality sequences unsuitable for robust analysis.
While these canonical benchmarks are widely adopted, they primarily consist of frontal, well-lit, human subjects captured under relatively controlled recording conditions. Such characteristics, although beneficial for training and in-domain evaluation, may underrepresent challenging motion dynamics, occlusions, cross-style variation, and distribution shifts encountered in real-world deployment. Consequently, paradigm-level robustness and generalization behavior may remain insufficiently characterized under standard benchmark settings.

To address this limitation, we construct two additional evaluation subsets specifically designed to stress robustness and out-of-distribution generalization. The \textit{Wild} subset (60 videos) 
emphasizes scenarios with rapid head motion, occlusions, diverse speaking and singing styles, and strong emotional expressions. These conditions amplify temporal variability and motion amplitude, providing a more stringent test of dynamic stability and synchronization robustness.
The \textit{Avatar} subset (40 videos) 
includes photorealistic AI-generated faces, 2D and 3D animated characters, artistic avatars and sculptures, sketch-style renderings, and non-human subjects. This subset intentionally introduces substantial appearance distribution shifts to evaluate cross-style transferability and feature-level robustness.
This two-level dataset design, combining canonical benchmarks with robustness-oriented subsets, enables systematic evaluation under both controlled and distribution-shifted conditions. It facilitates analysis of paradigm-specific behavior not only in terms of in-domain performance but also under challenging motion, expression, and appearance variations. Such design is particularly important for temporally-aware evaluation, as distribution shifts often exacerbate temporal instability and alignment sensitivity that remain obscured under conventional in-domain benchmarks.

\begin{figure}[tbp]
    \centering
    \includegraphics[trim=0 0 0 0, clip=true, width=\linewidth]{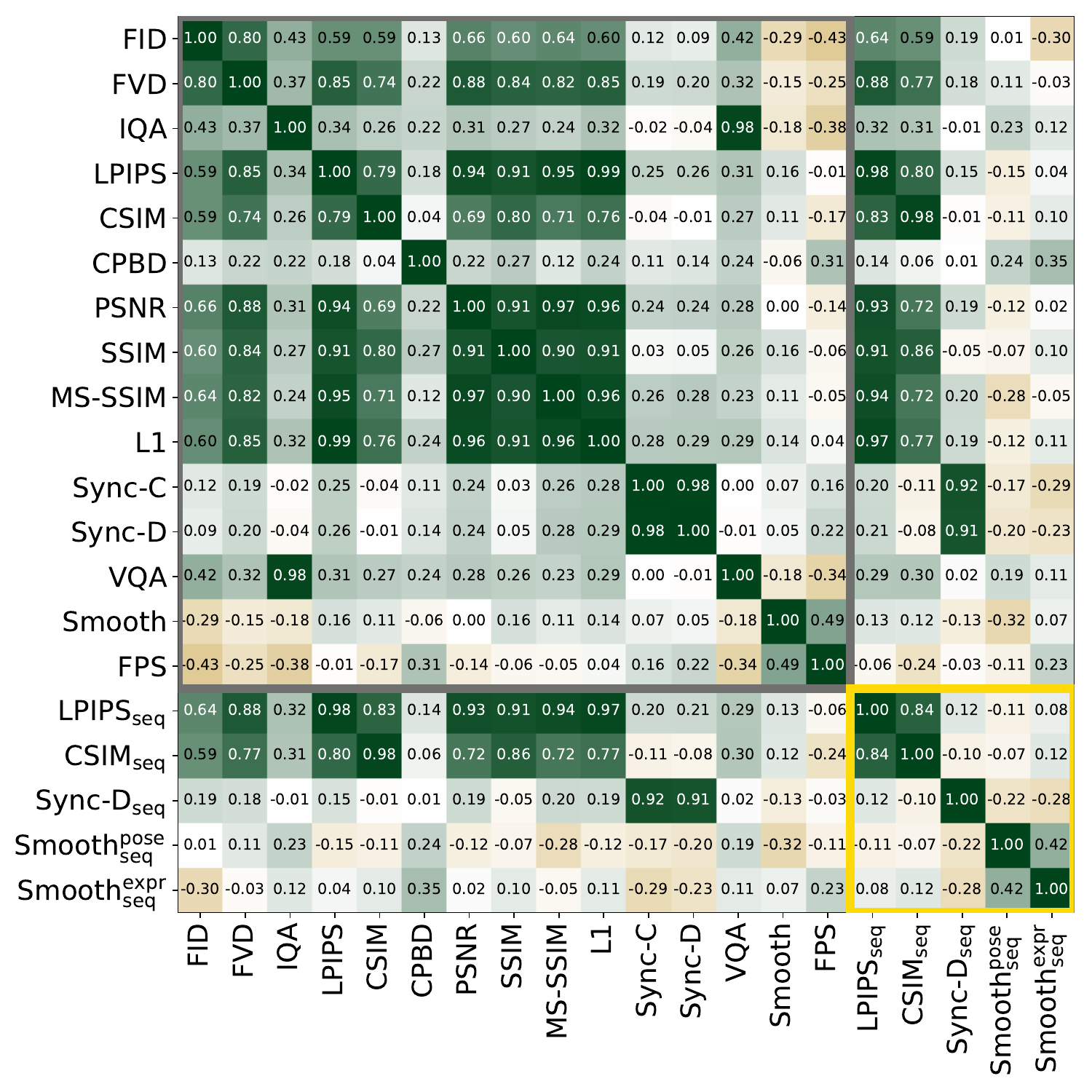}
    \caption{Spearman rank correlation matrix computed from normalized metric scores averaged across datasets. Each method is ranked per metric, and the correlation measures monotonic agreement. Frame-level metrics show strong correlation, while sequence-level metrics are less correlated, highlighting complementary temporal evaluation.
    The gray bounding box corresponds to conventional frame-level metrics, whereas the gold bounding box highlights the proposed Soft-DTW-based sequence-level metrics.
    }
    \label{fig:spearman-rank}
\end{figure}

\begin{figure}[!htbp]
    \centering
    \includegraphics[trim=0 0 0 0, clip=true, width=\linewidth]{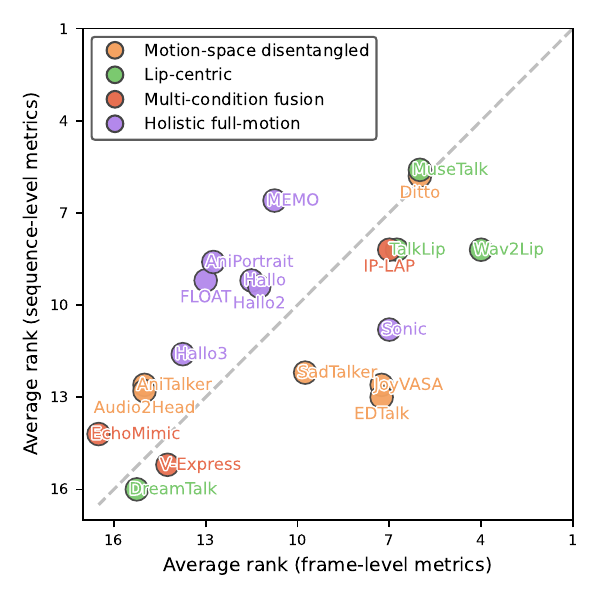}
    \caption{
    Ranking consistency between frame-level and sequence-level metrics. Each point represents a method, with the $x$-axis showing average rank across original metrics (LPIPS, CSIM, Sync-D, Smooth) and the $y$-axis showing average rank across sequence-level counterparts. 
    Lower ranks (closer to 1) indicate better performance.
    The dashed diagonal indicates perfect agreement; points above (below) show methods that improve (decline) under sequence-level evaluation. Points are color-coded by paradigm, highlighting both overall correlations and paradigm-specific shifts due to temporal alignment.
    }
    \label{fig:ranking_scatter}
\end{figure}

\subsection{Evaluation Framework}


Evaluating audio-driven talking-head generation requires metrics that capture both visual and temporal aspects. We adopt a unified, sequence-aware framework combining complementary metrics with standardized testing protocols. 

\textbf{Metric taxonomy and design rationale.} 
Audio-driven talking-head generation involves multiple quality aspects, which we capture using six complementary metric categories.

\emph{Perceptual visual quality} captures how realistic and visually convincing generated content appears. Metrics such as FID and FVD measure distributional similarity at image and video levels, respectively, quantifying global realism. LPIPS evaluates perceptual similarity between generated and reference frames in learned feature spaces. Image Quality Assessment (IQA) and Video Quality Assessment (VQA) assess perceptual quality at image and video scales, while Cumulative Probability of Blur Detection (CPBD) estimates perceptual sharpness. Together, these metrics provide a multi-scale, perceptually grounded assessment of visual fidelity, reflecting human judgments of realism and image quality.
\emph{Pixel-level reconstruction fidelity} evaluates low-level signal consistency between generated and ground-truth frames. Metrics such as PSNR, SSIM, MS-SSIM, and L1 capture structural and numerical reconstruction accuracy. While they do not necessarily align with perceptual realism, they remain essential for understanding model performance in preserving fine-grained details and structural consistency.

\emph{Identity preservation} is critical in talking-head generation to prevent semantic drift over time. CSIM measures cosine similarity in a high-dimensional face embedding space, quantifying whether the generated sequence maintains the subject’s identity consistently across frames. This semantic-level evaluation complements visual fidelity metrics, ensuring that identity is preserved even in the presence of expressive or extreme motions.
\emph{Audio-visual synchronization} assesses cross-modal alignment between speech signals and lip motion. Sync confidence (Sync-C) and Sync distance (Sync-D) quantify temporal coherence between audio and visual streams, capturing a defining characteristic of audio-driven methods. Accurate synchronization is necessary for perceptual realism and intelligibility, particularly in dialogue or singing scenarios.
\emph{Expression and motion naturalness} evaluates the plausibility of temporal dynamics. The Smooth metric measures continuity and stability in pose and expression trajectories, emphasizing biomechanical realism and natural evolution of motion. This category focuses on temporal dynamics that are not captured by appearance-based or pixel-level metrics, highlighting differences in modeling strategies for motion generation.
\emph{Computational efficiency} reflects practical deployment considerations. FPS measures inference speed, providing insight into real-time capability and resource requirements, which is critical for applications ranging from virtual avatars to telepresence.

\subsection{Temporally-Aware Evaluation Metrics}
\label{sec:metrics}
\textbf{Rethinking evaluation.
} While existing metrics capture complementary aspects of video quality, they assume strict frame-to-frame correspondence. This makes them overly sensitive to small temporal offsets, phase shifts, or speaking-rate variations, which can distort comparisons and obscure meaningful differences between methods. 
These observations suggest that evaluation for talking-head generation should satisfy two fundamental properties: (i) \emph{temporal order preservation}, ensuring that motion remains causally consistent; and (ii) \emph{temporal elasticity}, allowing flexible alignment under perceptually tolerable timing variations. Conventional frame-wise metrics satisfy neither property. We now present our temporally-aware evaluation metrics.

\textbf{From frame alignment to trajectory alignment.} Conventional frame-wise metrics assume strict temporal correspondence and equal sequence lengths. 
Given a generated video $\mX = \{\vx_t\}_{t=1}^{T}$ and a reference video $\mY = \{\vy_t\}_{t=1}^{T}$, frame-wise evaluation computes the average feature distance:
\begin{equation}
\text{Dist}_{\text{frame}}(\mX, \mY) = \frac{1}{T} \sum_{t=1}^{T} d\big(\phi(\vx_t), \phi(\vy_t)\big),
\end{equation}
where $\phi(\cdot)$ is a task-specific feature extractor (\eg, perceptual, identity, synchronization, or motion embeddings) and $d(\cdot, \cdot)$ is a feature-space distance, typically the squared $\ell_2$ norm. 
While simple, this formulation is highly sensitive to minor temporal offsets, phase shifts, or rhythm variations commonly encountered in generated talking-head videos.

To overcome these limitations, we treat each video as a trajectory in feature space and compute a sequence-level similarity using Soft-DTW~\cite{cuturi2017soft}. 
Let $\mF = \{\phi(\vx_t)\}_{t=1}^{T}$, $\mG = \{\phi(\vy_s)\}_{s=1}^{S}$
denote the feature sequences of the generated and reference videos. Soft-DTW computes a differentiable, globally-aligned distance between sequences:
\begin{equation}
\!\!\!\!\text{Soft-DTW}_\gamma(\mF, \mG) \!=\! -\gamma \log \!\sum_{\pi \in \Pi} \!\exp\! \Bigg(\!\!\! - \!\frac{1}{\gamma} \!\!\sum_{(t,s) \in \pi} \!\!d(\mF_t, \mG_s) \!\!\Bigg),
\label{eq:dtw}
\end{equation}
where $\Pi$ is the set of all monotonic alignment paths, and $\gamma > 0$ controls the trade-off between strict alignment ($\gamma \to 0$) and flexible matching ($\gamma$ large). By considering all possible monotonic alignments, Soft-DTW is robust to minor temporal misalignments, capturing global temporal coherence and sequence-level structure across the video.

To ensure comparability across sequences of different lengths, we define a normalized distance:
\begin{equation}
\text{Dist}_{\text{seq}}(\mX, \mY) = \frac{\text{Soft-DTW}_\gamma(\mF, \mG)}{T_{\max}},
\end{equation}
where $T_{\max} = \max(T, S)$. This formulation treats each video as a continuous trajectory in feature space, preserving temporal order while allowing non-linear alignment. 
This sequence-level perspective drives our evaluation framework.



\begin{figure}[tbp]
    \centering
    \resizebox{0.9\columnwidth}{!}{\includegraphics{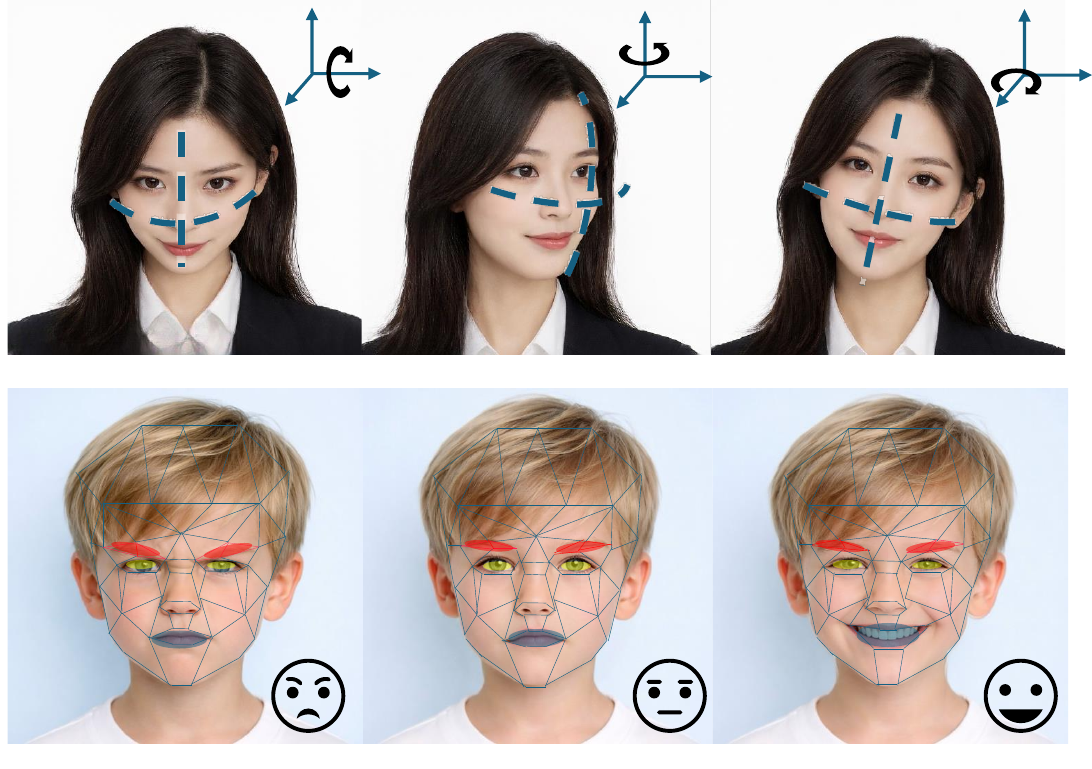}}
\caption{Head pose and expression control representation. The top row shows head pose variations via three rotations: pitch (up-down about the lateral axis), yaw (left-right about the vertical axis), and roll (in-plane about the longitudinal axis). The bottom row illustrates expression control using a 63-dimensional latent keypoint representation (21 3D keypoints) for modeling non-rigid facial deformation. Red regions mainly affect eyebrow motion, while blue regions control mouth expressions. Adjusting keypoints enables diverse facial expressions.
    }
    \label{fig:FACE}
\end{figure}

\textbf{Dimension-specific instantiations.} Within this unified framework, different evaluation aspects correspond to different feature extractors $\phi(\cdot)$.
For \emph{visual perceptual quality}, $\phi$ maps each frame to LPIPS feature embeddings extracted from pretrained perceptual networks\cite{zhang2018unreasonable}. Sequence alignment in this embedding space measures global perceptual consistency while tolerating minor temporal phase shifts.
For \emph{identity preservation}, $\phi$ corresponds to $\ell_2$-normalized facial embeddings from a pretrained face recognition model (\eg, InsightFace\cite{deng2019arcface}). Trajectory alignment captures identity stability across time, mitigating sensitivity to transient pose changes or occasional detection noise.
For \emph{audio-visual synchronization}, $\phi$ produces paired audio and visual embeddings from a pretrained synchronization network (\eg, SyncNet\cite{chung2016out}). Soft-DTW alignment between audio and visual embedding sequences quantifies global audiovisual coherence under variable speaking rates and small asynchronies.
For \emph{expression and motion naturalness}, we decompose facial motion into two separate feature trajectories: rigid head pose and non-rigid expression deformation (see Fig. \ref{fig:FACE}). Specifically, for each video frame, a 70-dimensional motion feature vector is extracted from the motion encoder\cite{guo2024liveportrait}:
\renewcommand{\labelenumi}{\roman{enumi}.}
\begin{enumerate}
    \item \textit{Head pose motion (7D)}: The first seven dimensions model global rigid head movement, consisting of three rotation angles (pitch, yaw, roll; 3D), one global scaling factor (1D), and a 3D translation vector $\vt = (t_x, t_y, t_z)$ representing spatial displacement. These parameters form a 7D rigid transformation capturing 3D orientation, scale, and position relative to the camera.
    \item \textit{Expression deformation (63D)}: The remaining 63 dimensions encode non-rigid facial deformations via 21 keypoints, each with 3D coordinates $(X, Y, Z)$. These keypoints span facial regions including the forehead, eyebrows, cheeks, eyes, nose, mouth, chin, and neck, providing an interpretable representation of expression dynamics such as eyebrow raising, lip protrusion, eye gaze shifts, and head or neck motion.
\end{enumerate}

Let the pose trajectories be $\mP = \{\phi_{\text{pose}}(\vx_t)\}_{t=1}^{T}$ and $\mP^{\text{ref}} = \{\phi_{\text{pose}}(\vy_s)\}_{s=1}^{S}$, and the expression trajectories be $\mK = \{\phi_{\text{expr}}(\vx_t)\}_{t=1}^{T}$ and $\mK^{\text{ref}} = \{\phi_{\text{expr}}(\vy_s)\}_{s=1}^{S}$, representing the pose and expression sequences extracted from the generated video $\mX$ and reference video $\mY$, respectively. Each trajectory is independently aligned using Soft-DTW:
\begin{align}
    \text{Dist}^{\text{pose}}_{\text{seq}}(\mX, \mY) &= \frac{\text{Soft-DTW}_\gamma(\mP, \mP^{\text{ref}})}{T_{\max}},\\
    \text{Dist}^{\text{expr}}_{\text{seq}}(\mX, \mY) &= \frac{\text{Soft-DTW}_\gamma(\mK, \mK^{\text{ref}})}{T_{\max}}.
\end{align}

Separating pose and expression allows interpretable evaluation of global head motion and fine-grained facial dynamics while remaining robust to frame rate variations, minor temporal offsets, and local jitter. This trajectory-level analysis provides physically meaningful, temporally-aware measures of motion smoothness and expression naturalness.

\begin{figure}[tbp]
    \centering
    \subfloat[CSIM]{%
        \includegraphics[width=0.33\linewidth]{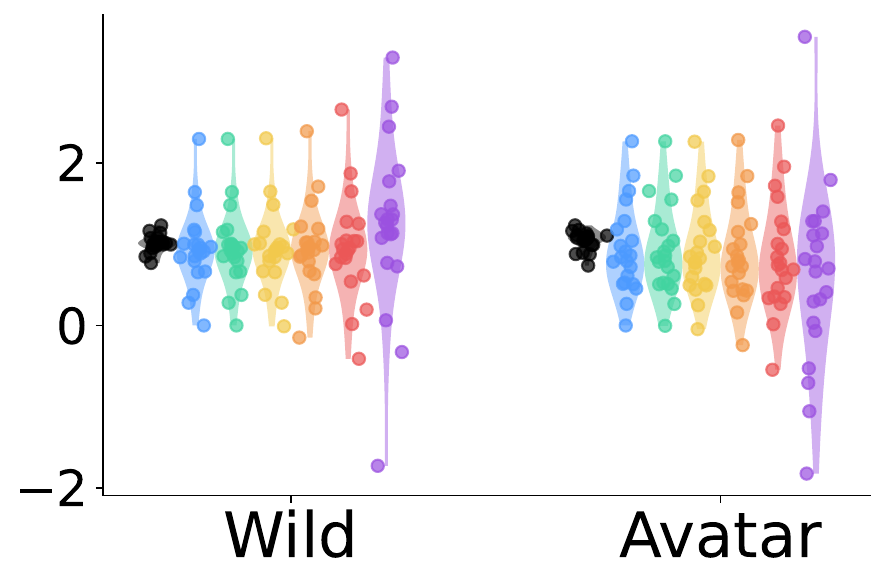}}
    \subfloat[LPIPS]{%
        \includegraphics[width=0.33\linewidth]{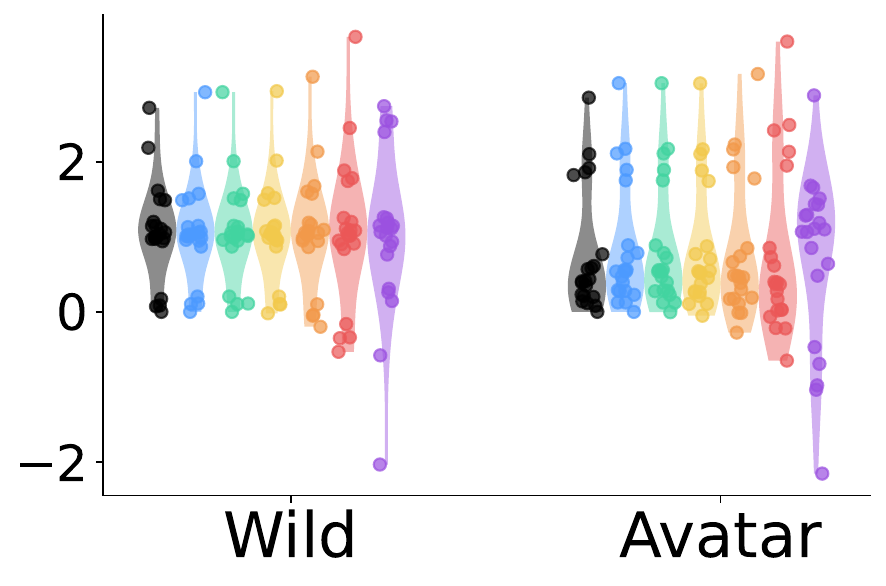}}
    \subfloat[Sync-D]{%
        \includegraphics[width=0.33\linewidth]{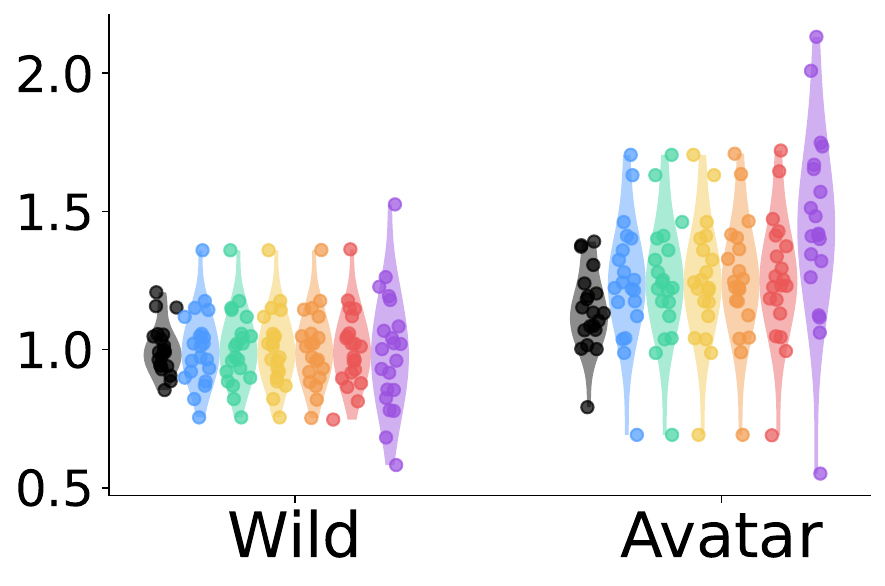}}\\
    \subfloat[Smooth$^{\text{pose}}$]{%
        \includegraphics[width=0.33\linewidth]{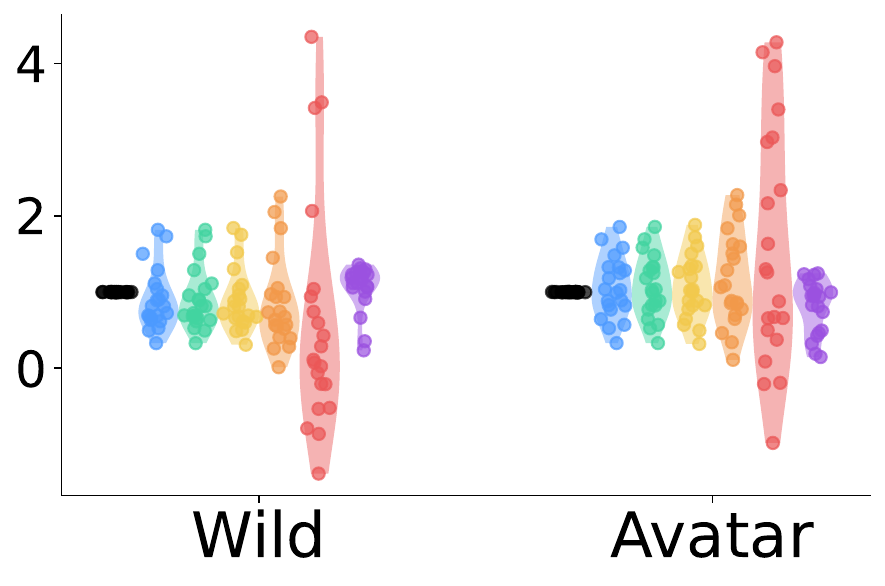}}
    \subfloat[Smooth$^{\text{expr}}$]{%
        \includegraphics[width=0.33\linewidth]{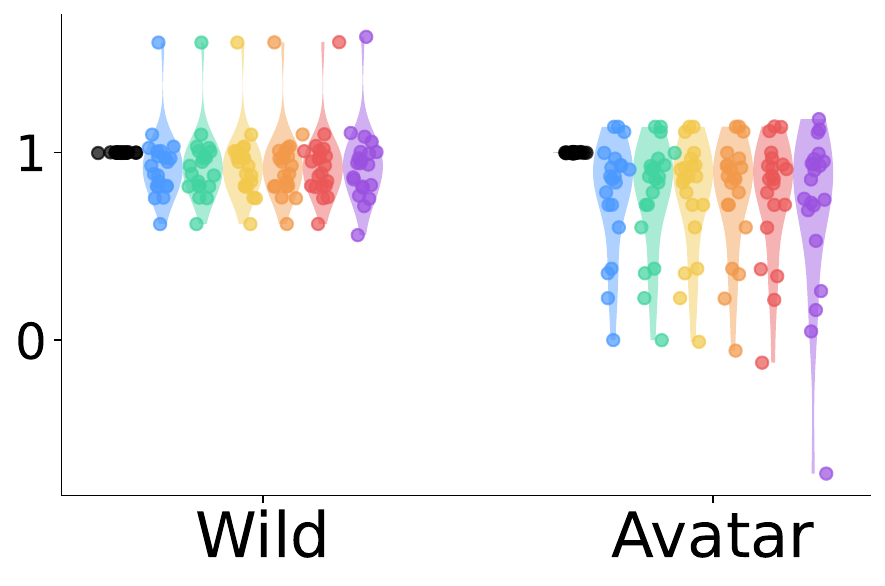}}
    \raisebox{0pt}[0pt][0pt]{%
    \includegraphics[width=0.33\linewidth]{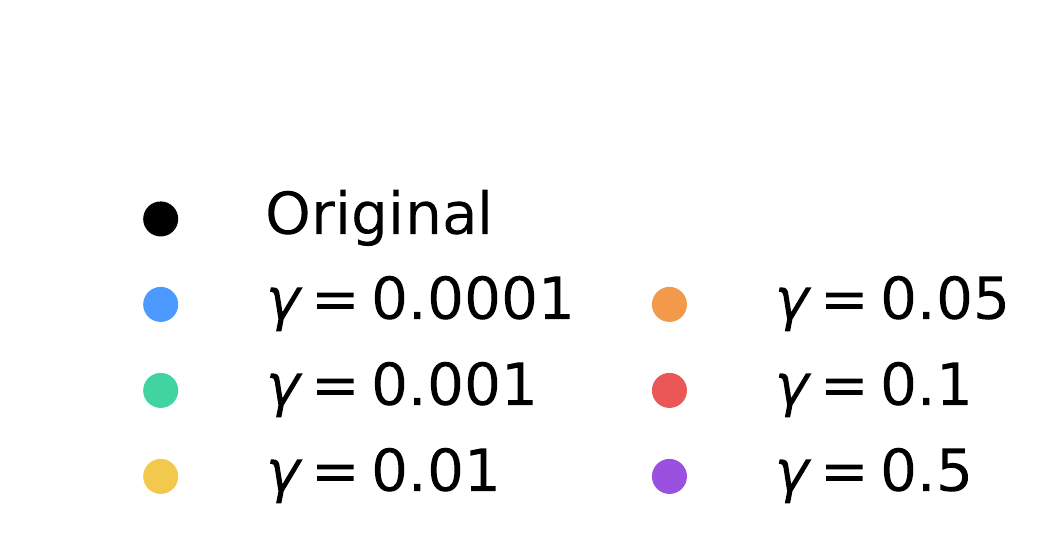}%
}
    \caption{
    Comparison of baseline and Soft-DTW variants under different $\gamma$ values across multiple evaluation metrics. Since the metrics exhibit different numerical scales, all values are normalized by their mean for fair comparison. Original denotes the metric computed without Soft-DTW alignment.
    }
    \label{fig:gamma_all}
\end{figure}

By unifying heterogeneous frame-wise metrics into a trajectory-alignment framework, our approach shifts evaluation from isolated frames to full-sequence analysis, better reflecting perceptual realism in talking-head videos. 
In essence, temporally-aware trajectory alignment provides a principled foundation for evaluating dynamic generation. 
Fig. \ref{fig:spearman-rank} shows the Spearman rank correlation matrix, computed from normalized metric scores averaged across the seven datasets. Each method is ranked per metric, and the correlation measures monotonic agreement, with higher values indicating more consistent rankings across methods. Sequence-level metrics show weaker correlations with frame-level metrics, highlighting their complementary evaluation of temporal dynamics.

\begin{figure*}[tbp]
    \centering
\subfloat[FID$\downarrow$]{\includegraphics[width=0.16\textwidth]{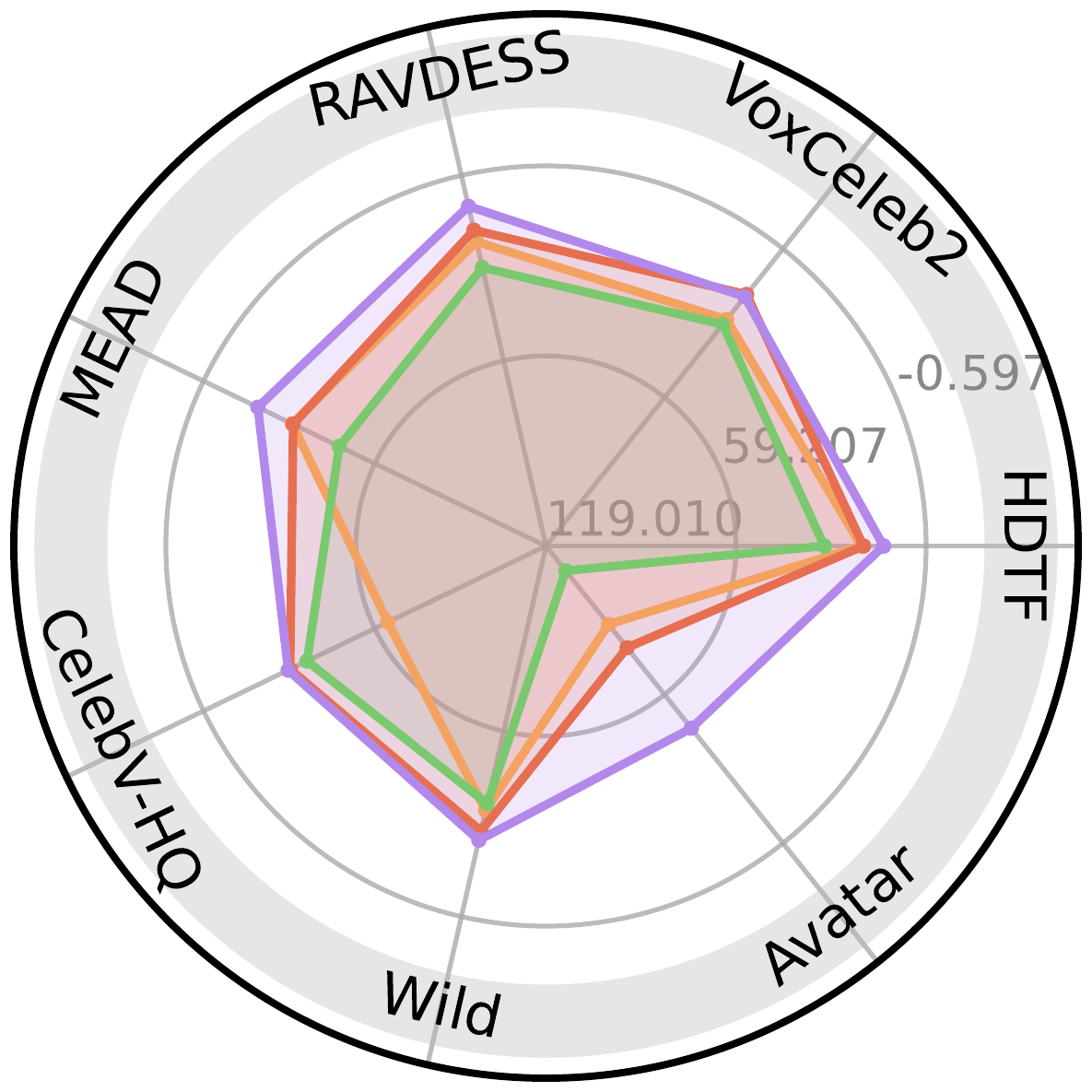}}
\subfloat[FVD$\downarrow$]{\includegraphics[width=0.16\textwidth]{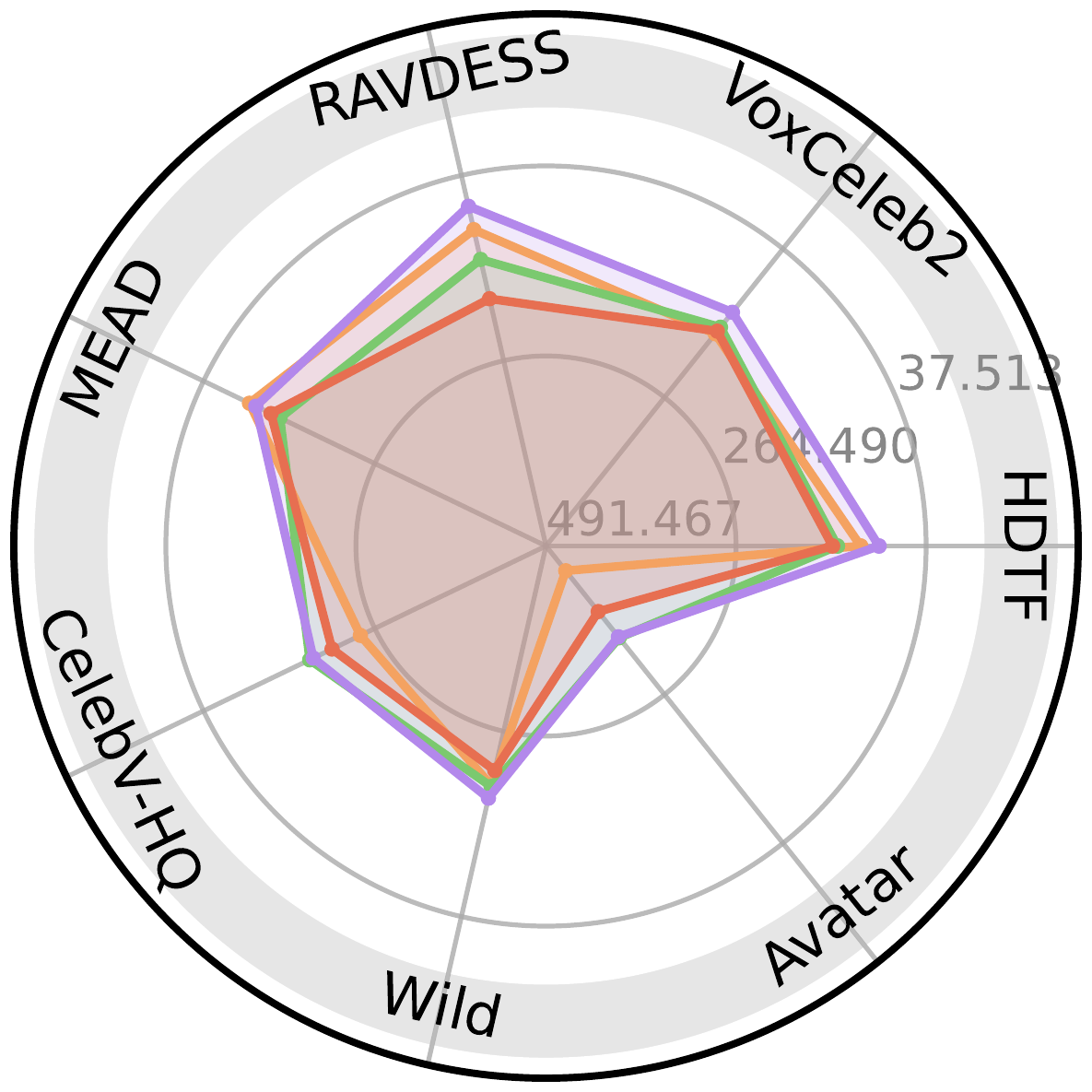}}
\subfloat[LPIPS$\downarrow$]{\includegraphics[width=0.16\textwidth]{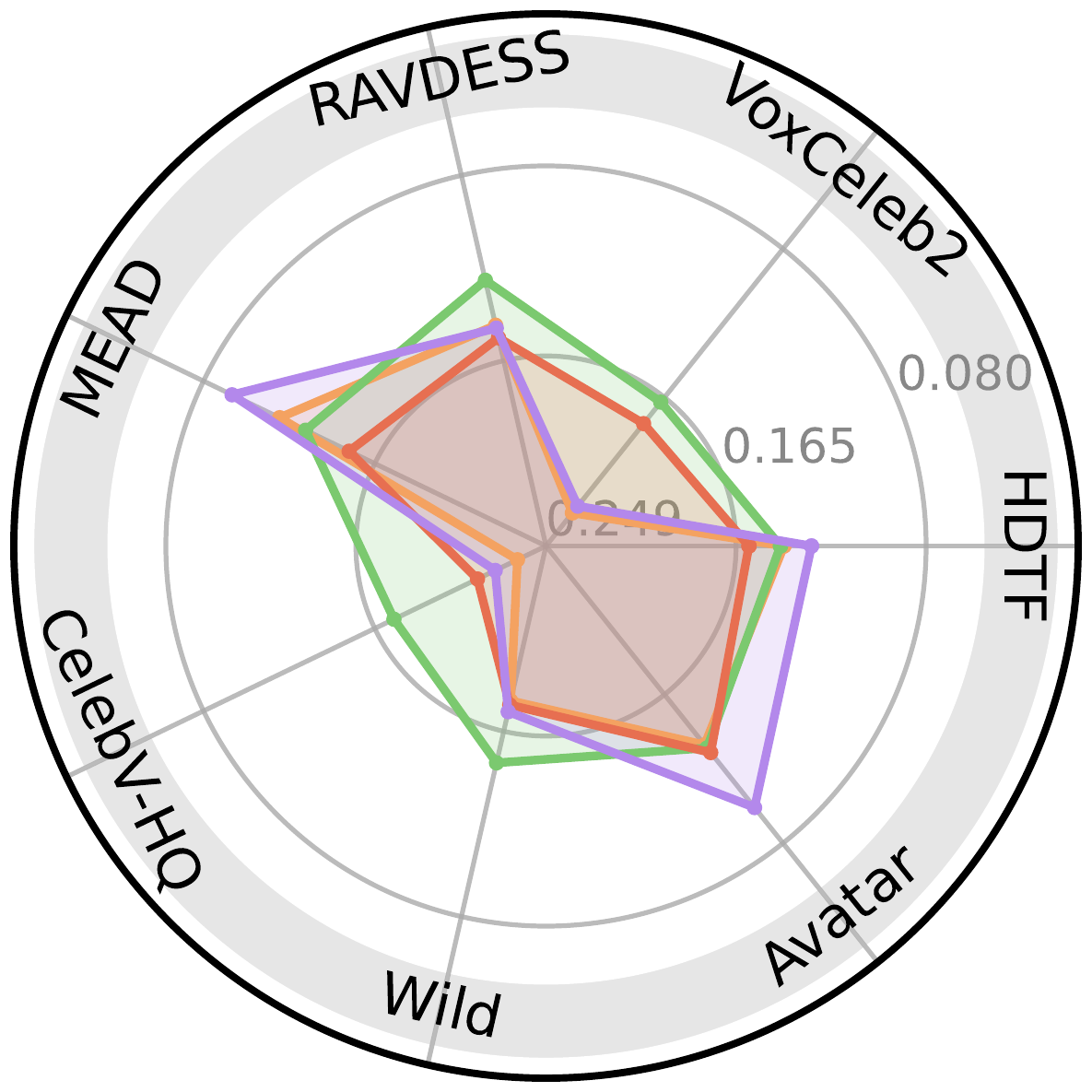}}
\subfloat[IQA$\uparrow$]{\includegraphics[width=0.16\textwidth]{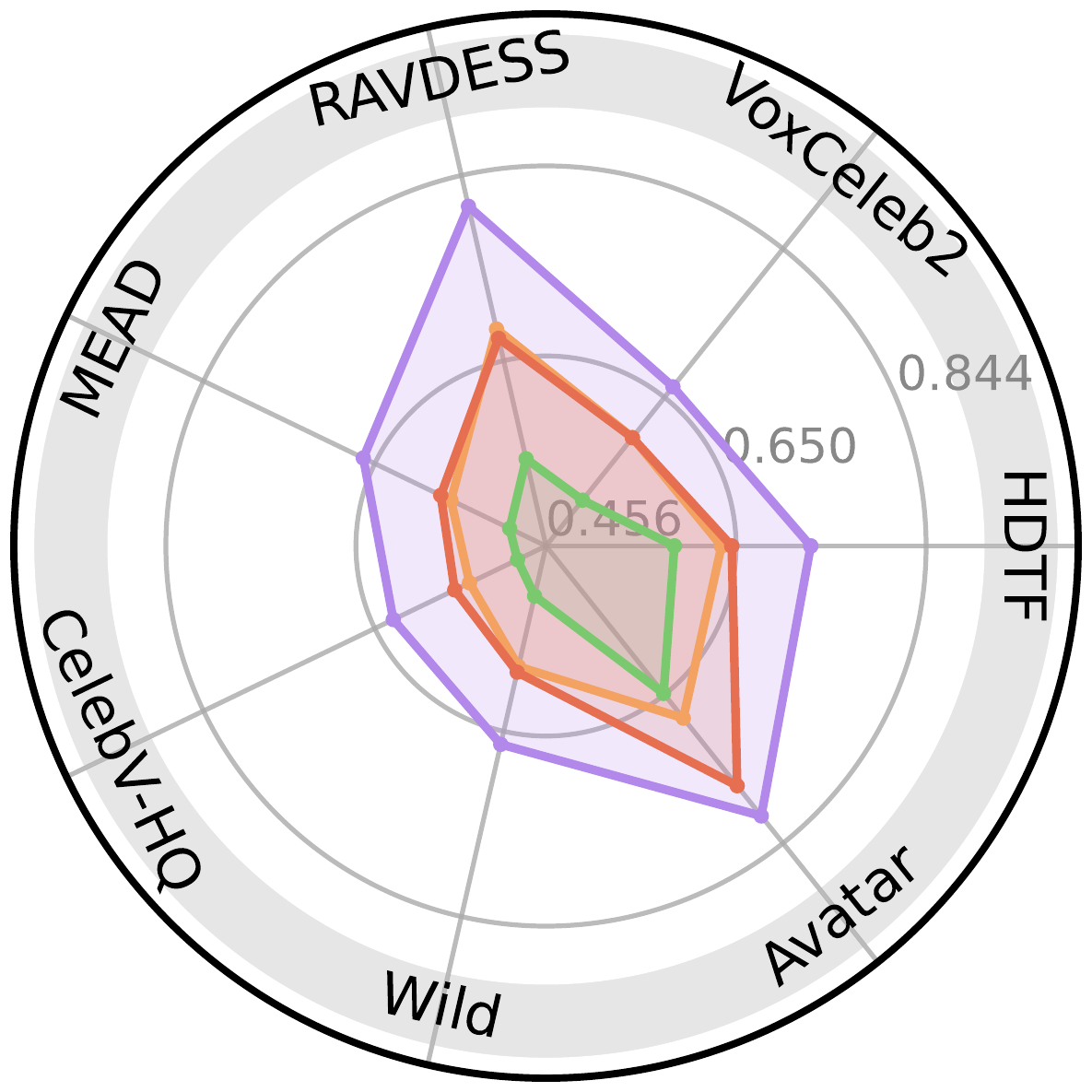}}
\subfloat[VQA$\uparrow$]{\includegraphics[width=0.16\textwidth]{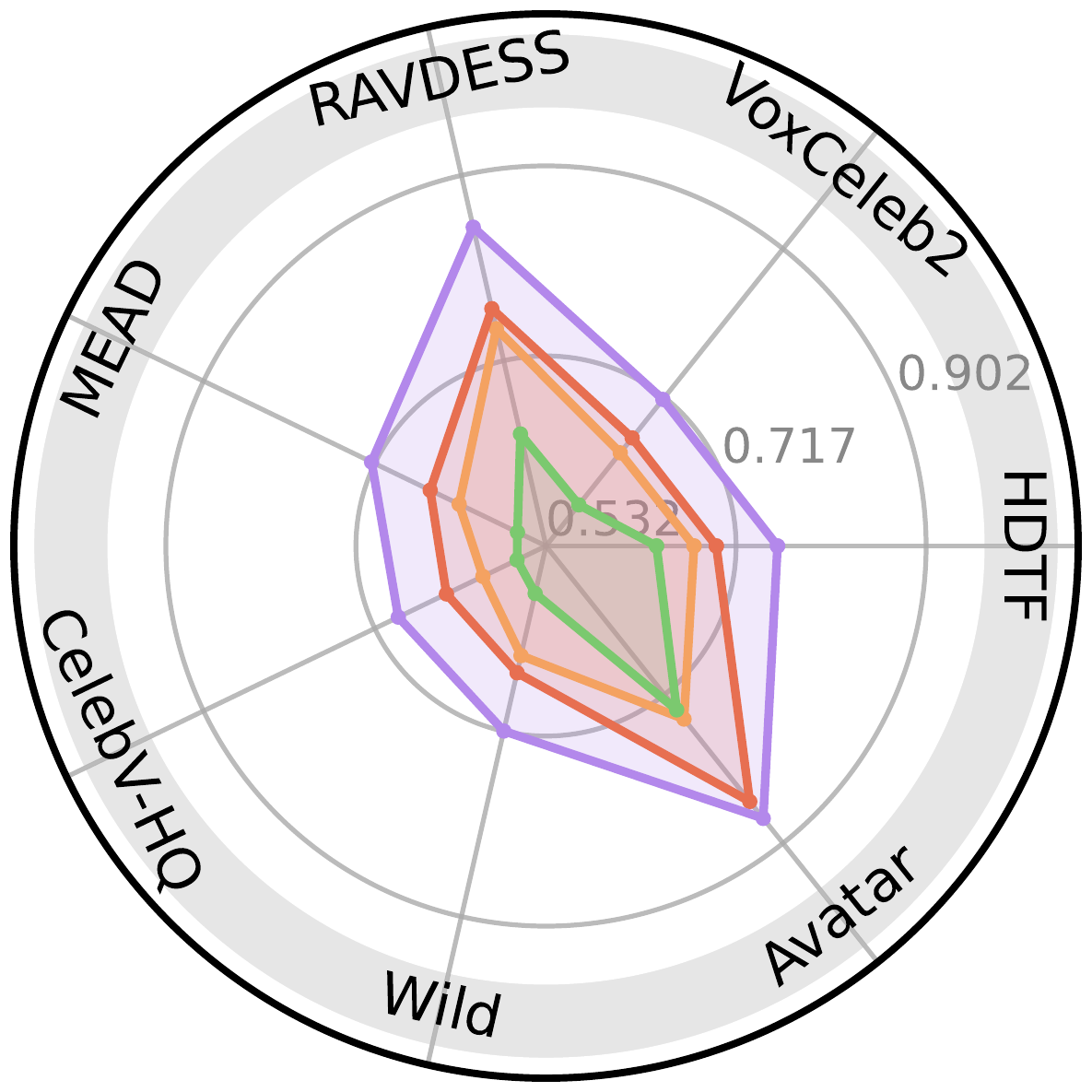}}
\subfloat[CPBD$\uparrow$]{\includegraphics[width=0.16\textwidth]{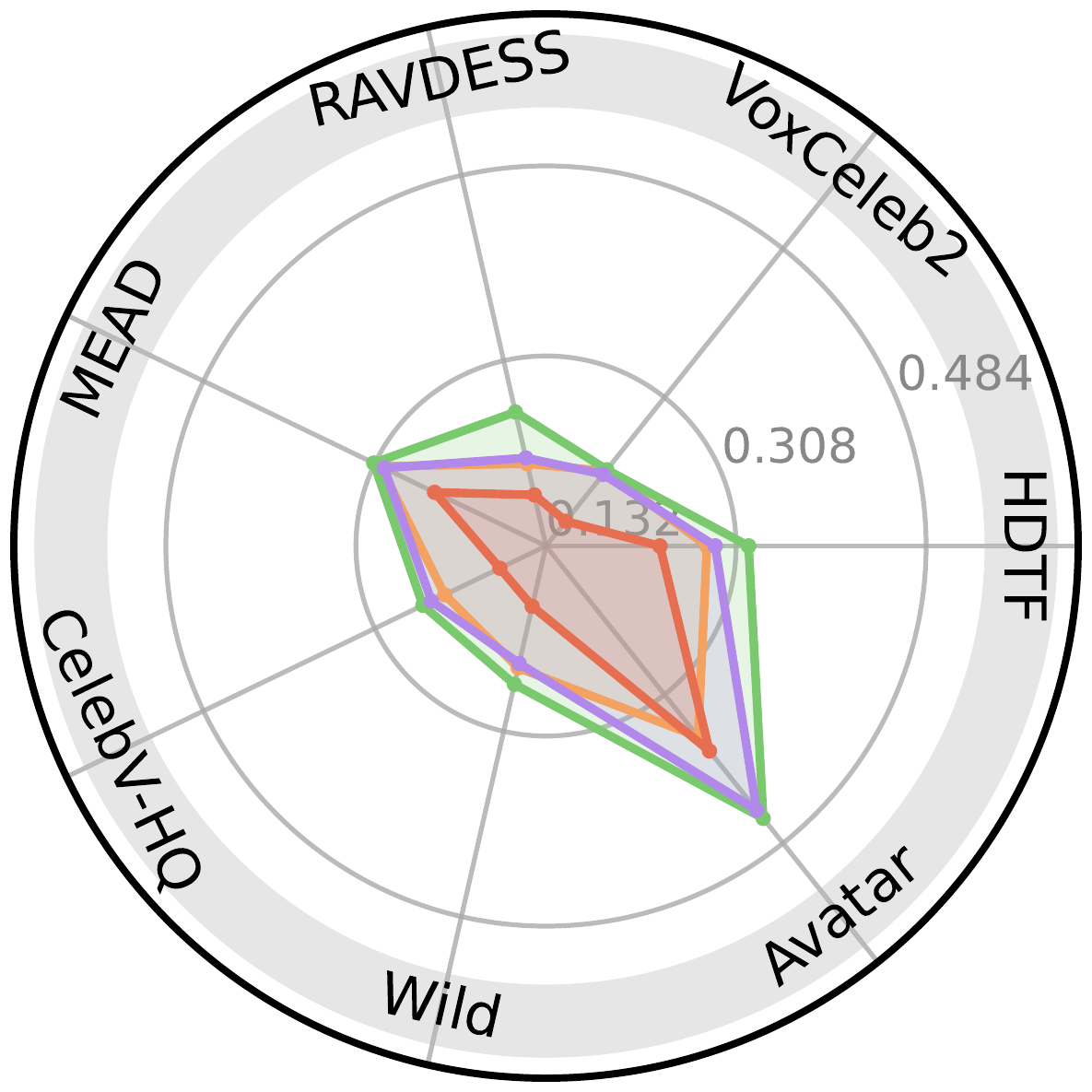}}\\
\subfloat[PSNR$\uparrow$]{\includegraphics[width=0.16\textwidth]{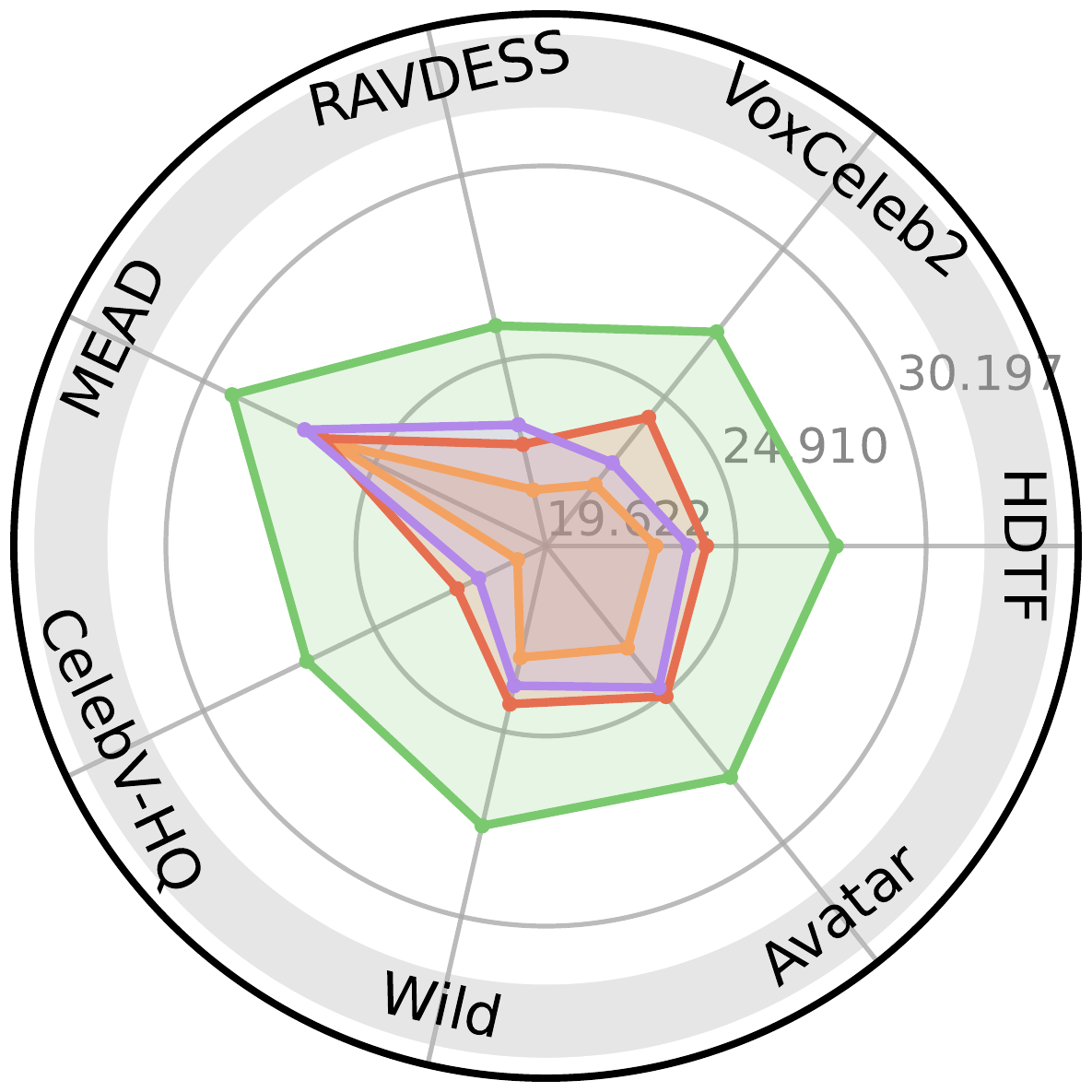}}
\subfloat[SSIM$\uparrow$]{\includegraphics[width=0.16\textwidth]{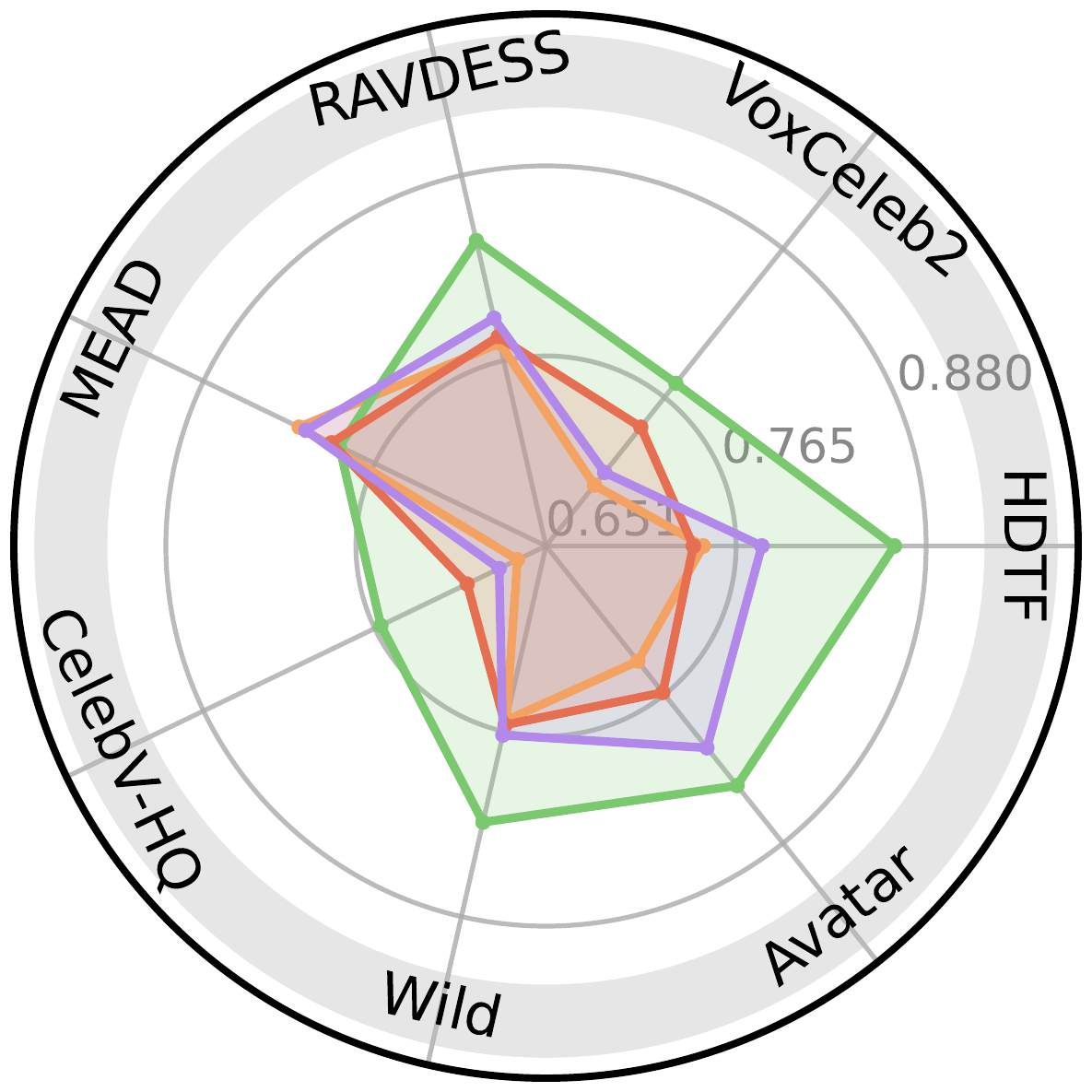}}
\subfloat[MS-SSIM$\uparrow$]{\includegraphics[width=0.16\textwidth]{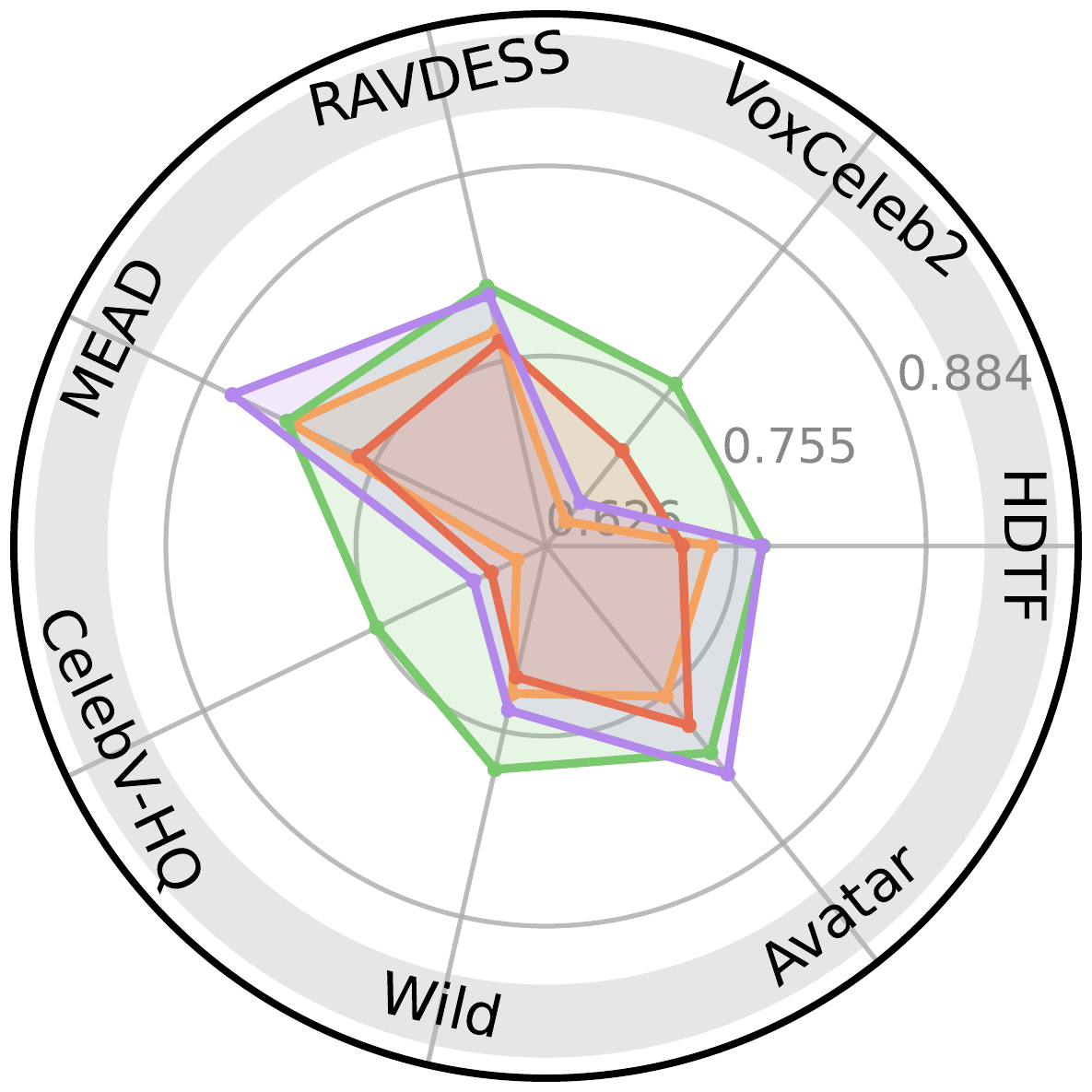}}
\subfloat[L1$\downarrow$]{\includegraphics[width=0.16\textwidth]{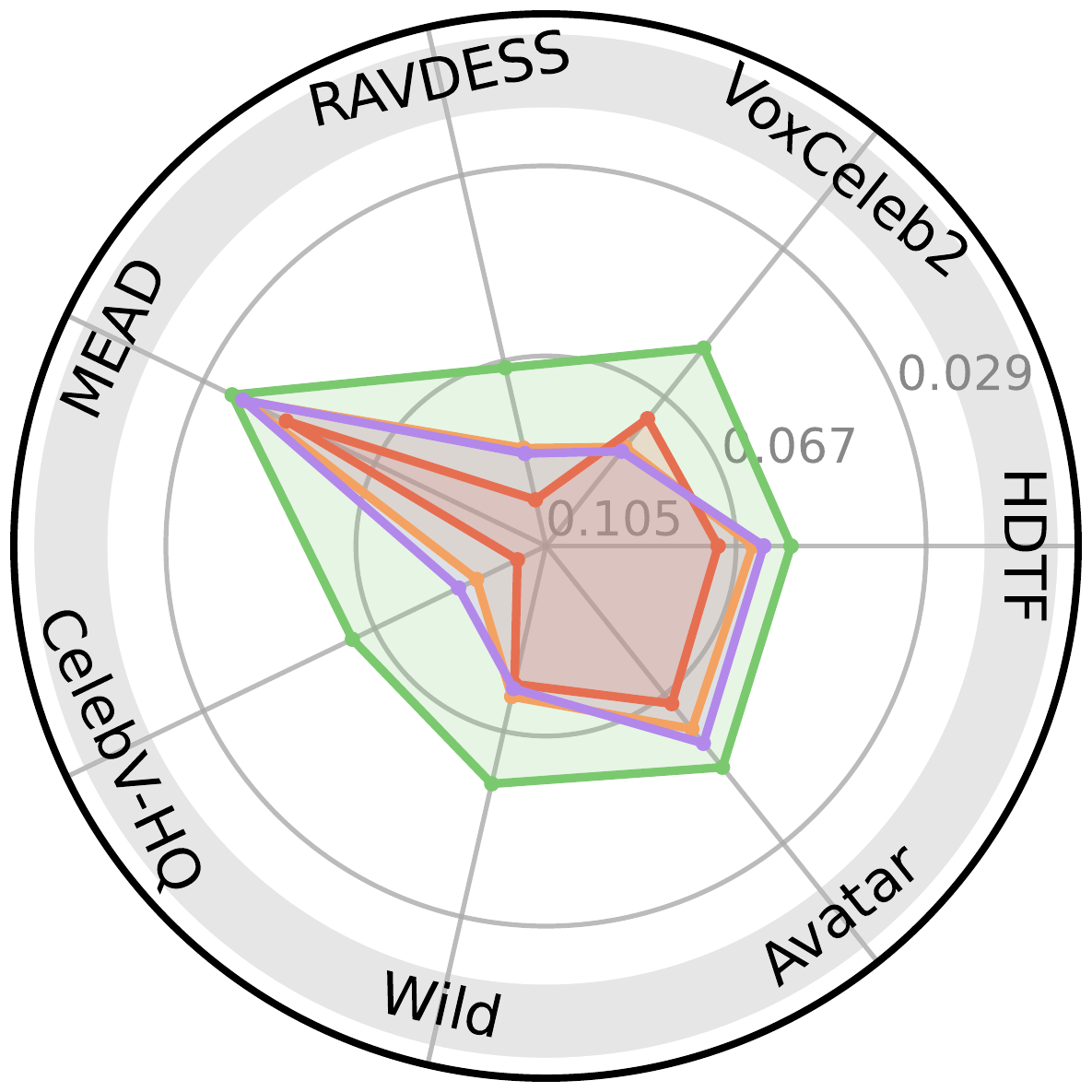}}
\hspace{0.05cm}
\raisebox{0pt}[0pt][0pt]{%
    \includegraphics[trim=5cm 1.8cm 5cm 1.8cm, clip, width=0.15\textwidth]{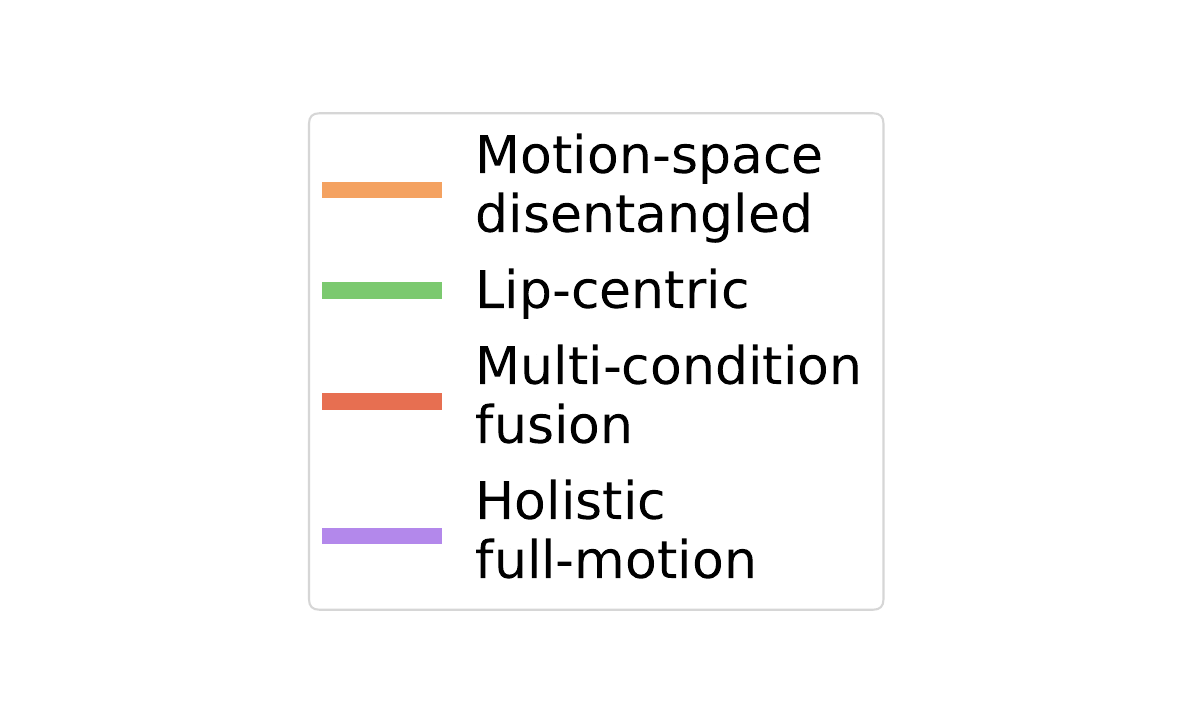}%
    
}
\captionsetup[subfloat]{labelformat=parens}
\subfloat[CSIM$\uparrow$]{\includegraphics[width=0.16\textwidth]{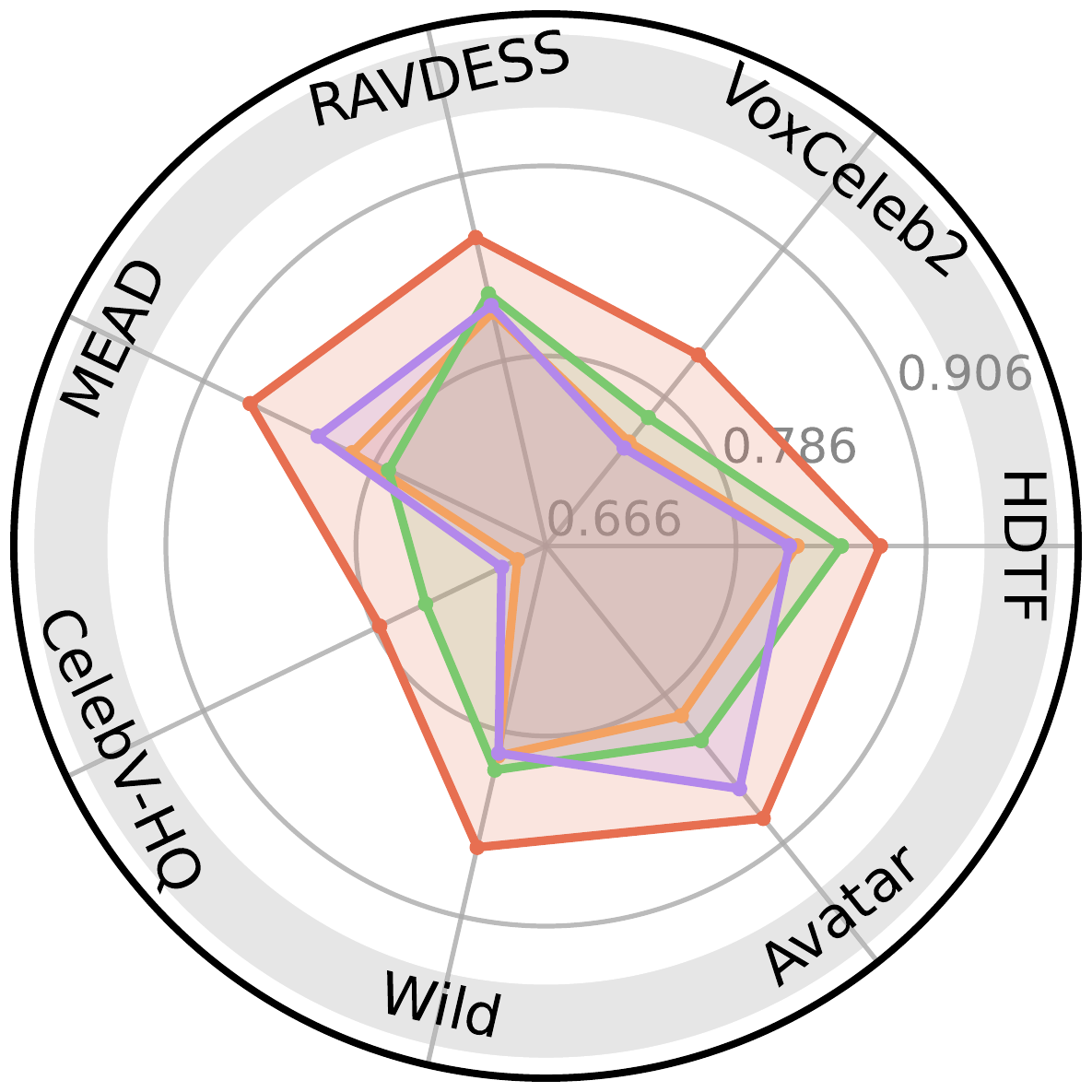}}
\subfloat[Sync-C$\uparrow$]{\includegraphics[width=0.16\textwidth]{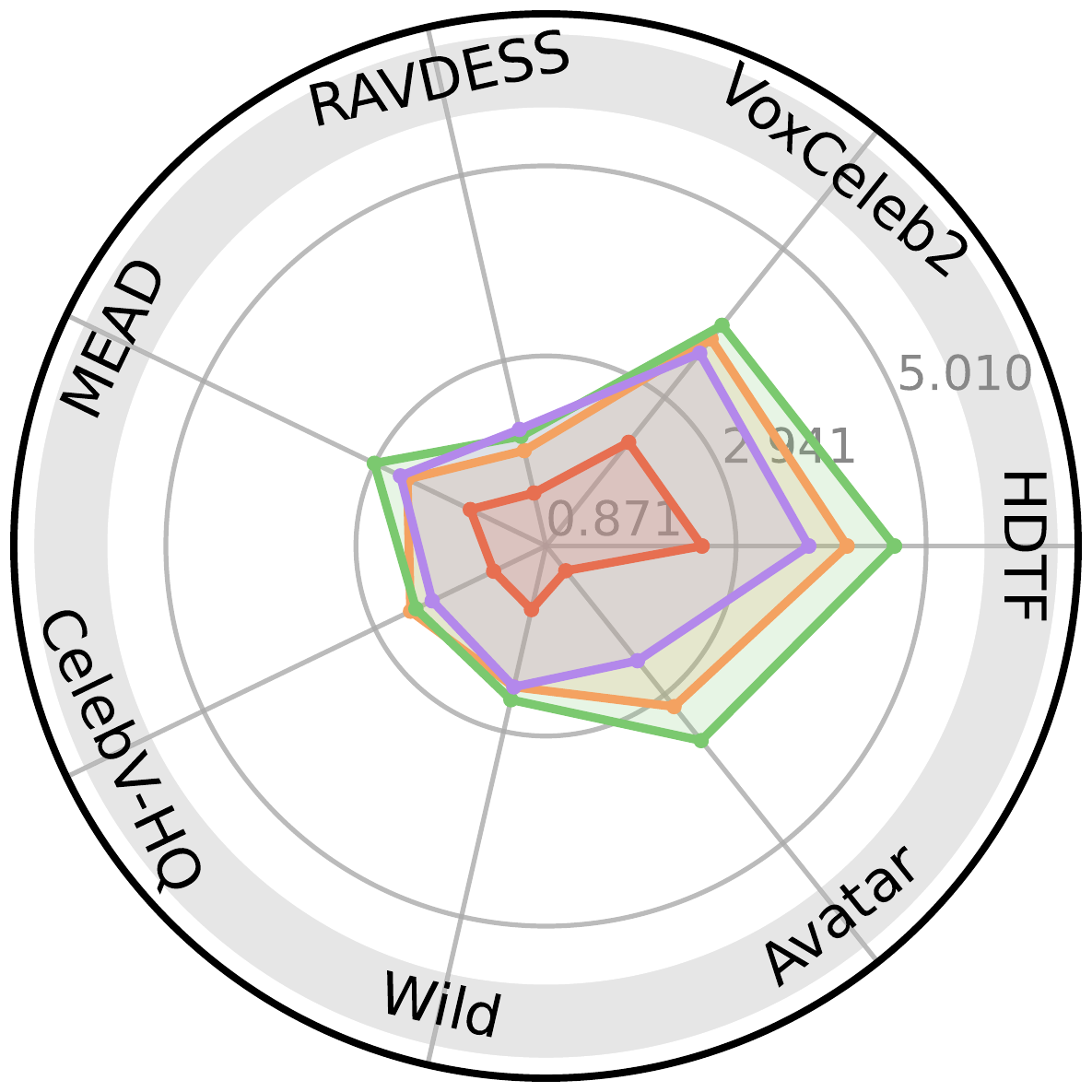}}
\subfloat[Sync-D$\downarrow$]{\includegraphics[width=0.16\textwidth]{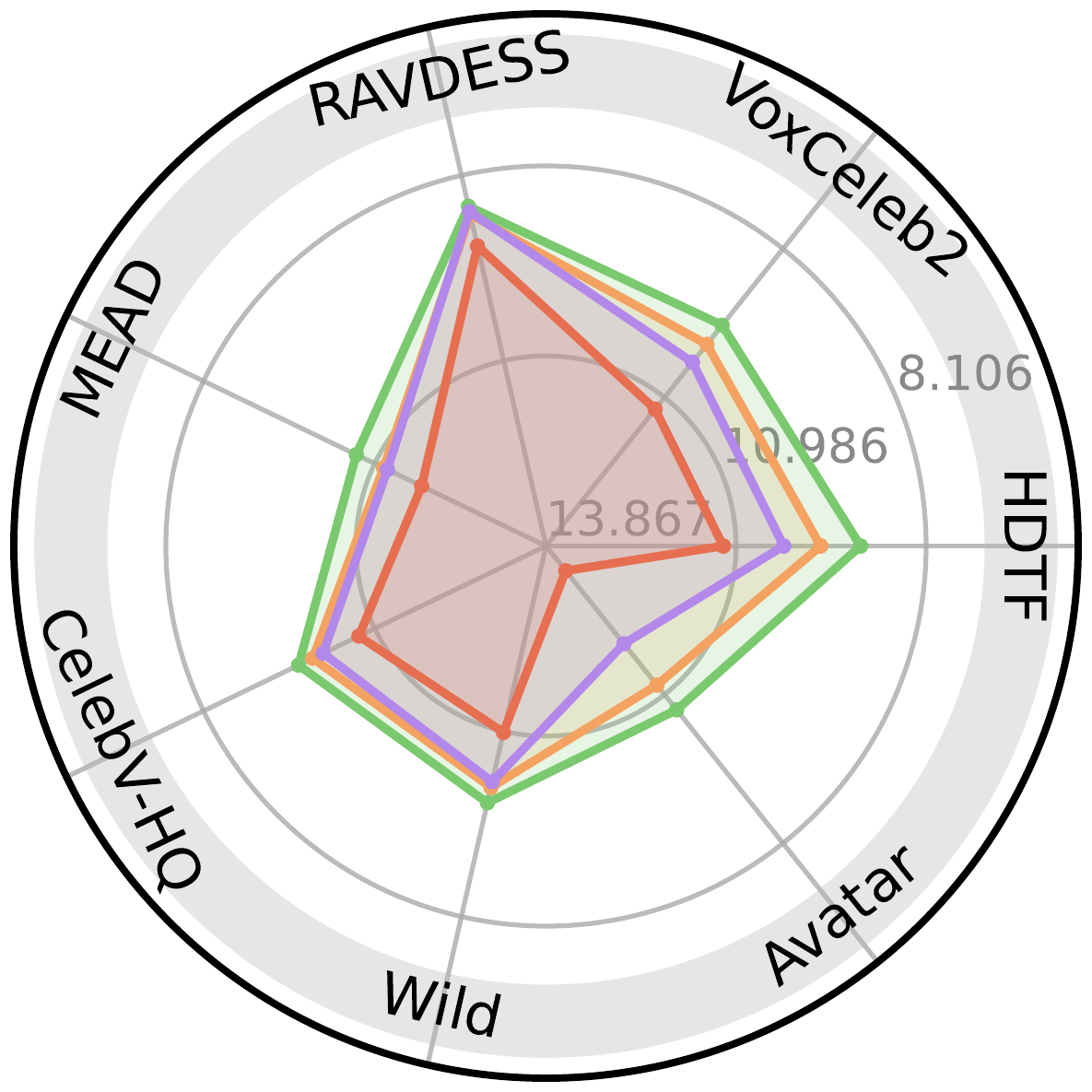}}
\subfloat[Smooth$\uparrow$]{\includegraphics[width=0.16\textwidth]{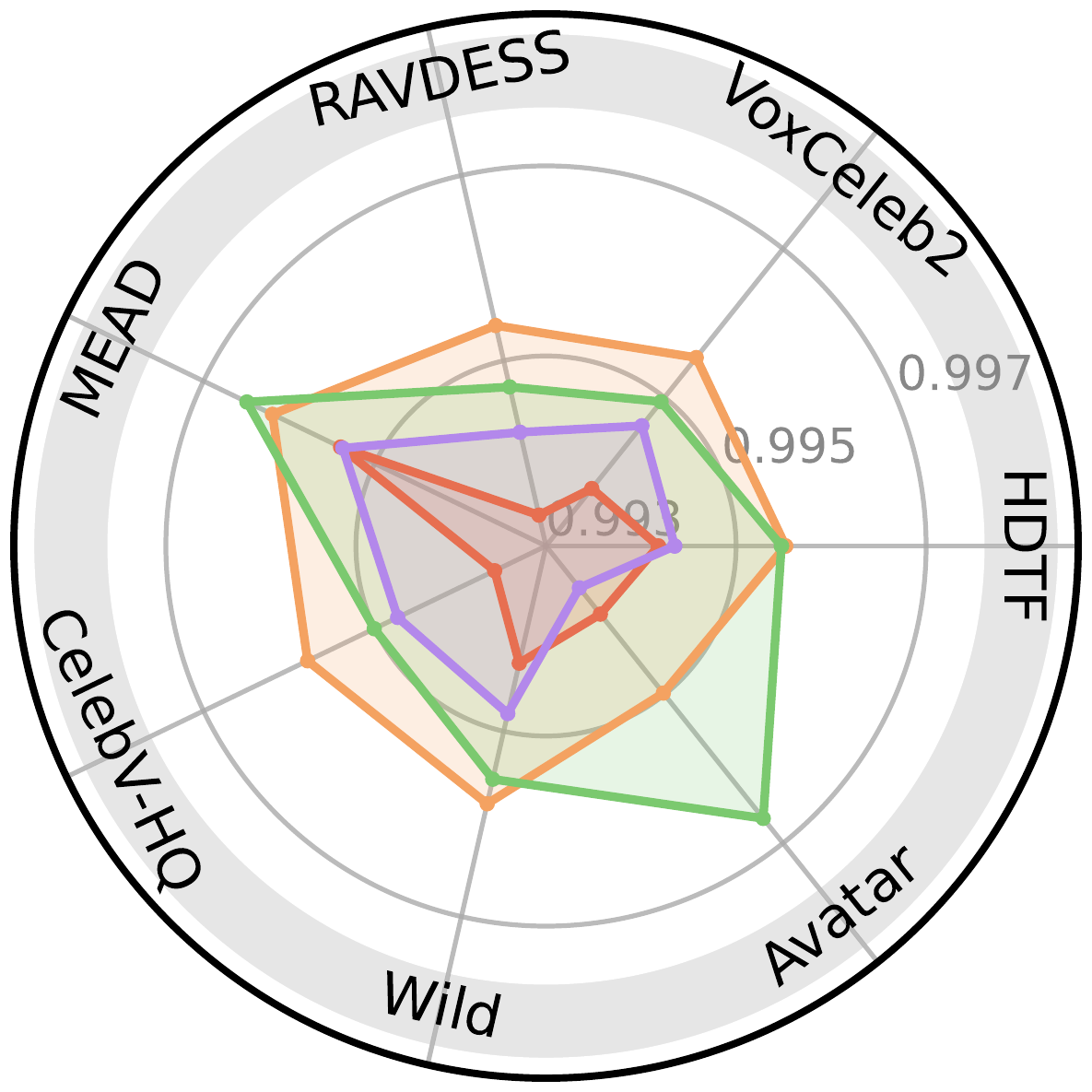}}
\subfloat[FPS$\uparrow$]{\includegraphics[width=0.16\textwidth]{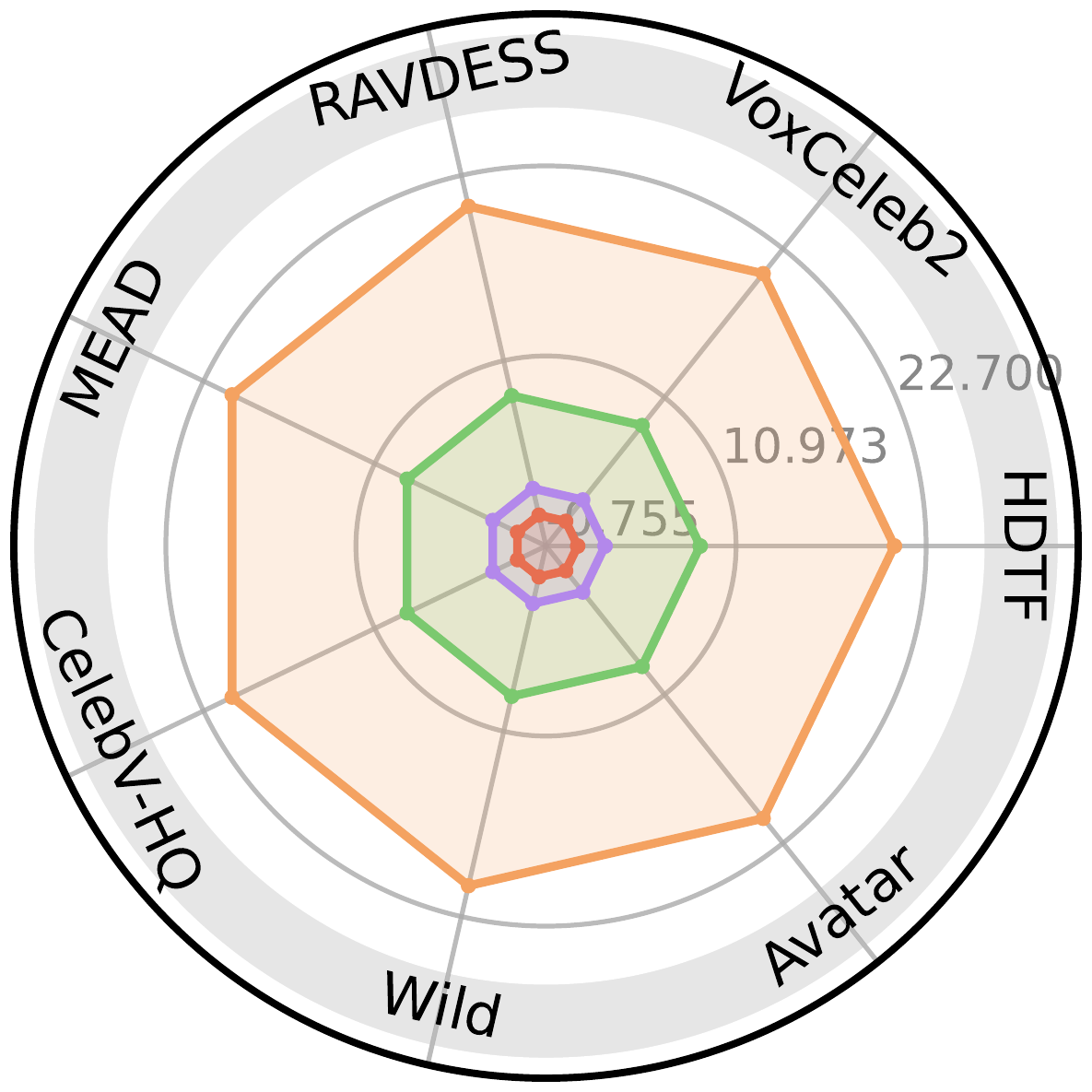}}\\
\subfloat[LPIPS$_{\mathrm{seq}}\downarrow$]{\includegraphics[width=0.16\textwidth]{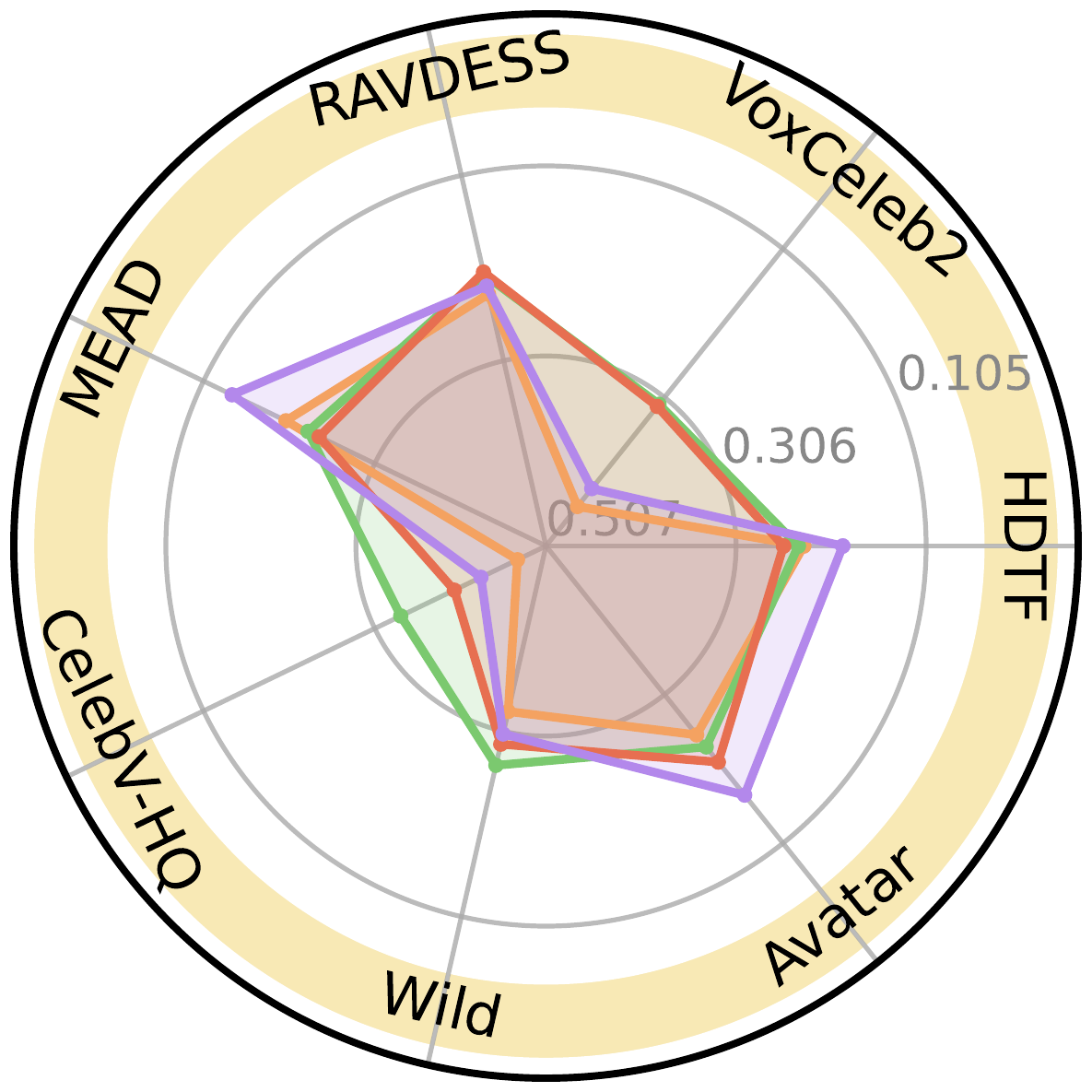}}
\subfloat[CSIM$_{\mathrm{seq}}\downarrow$]{\includegraphics[width=0.16\textwidth]{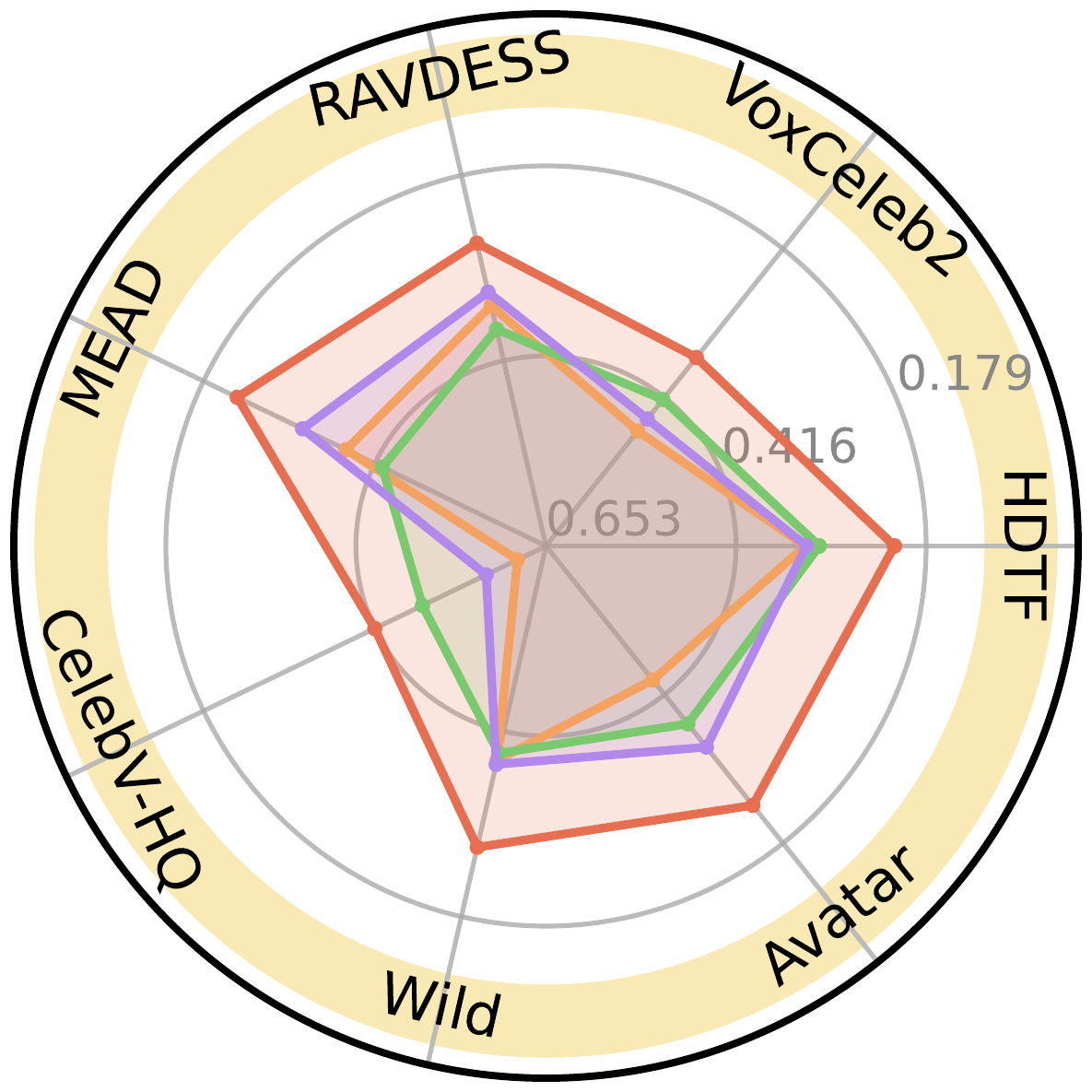}}
\subfloat[Sync-D$_{\mathrm{seq}}\downarrow$]{\includegraphics[width=0.16\textwidth]{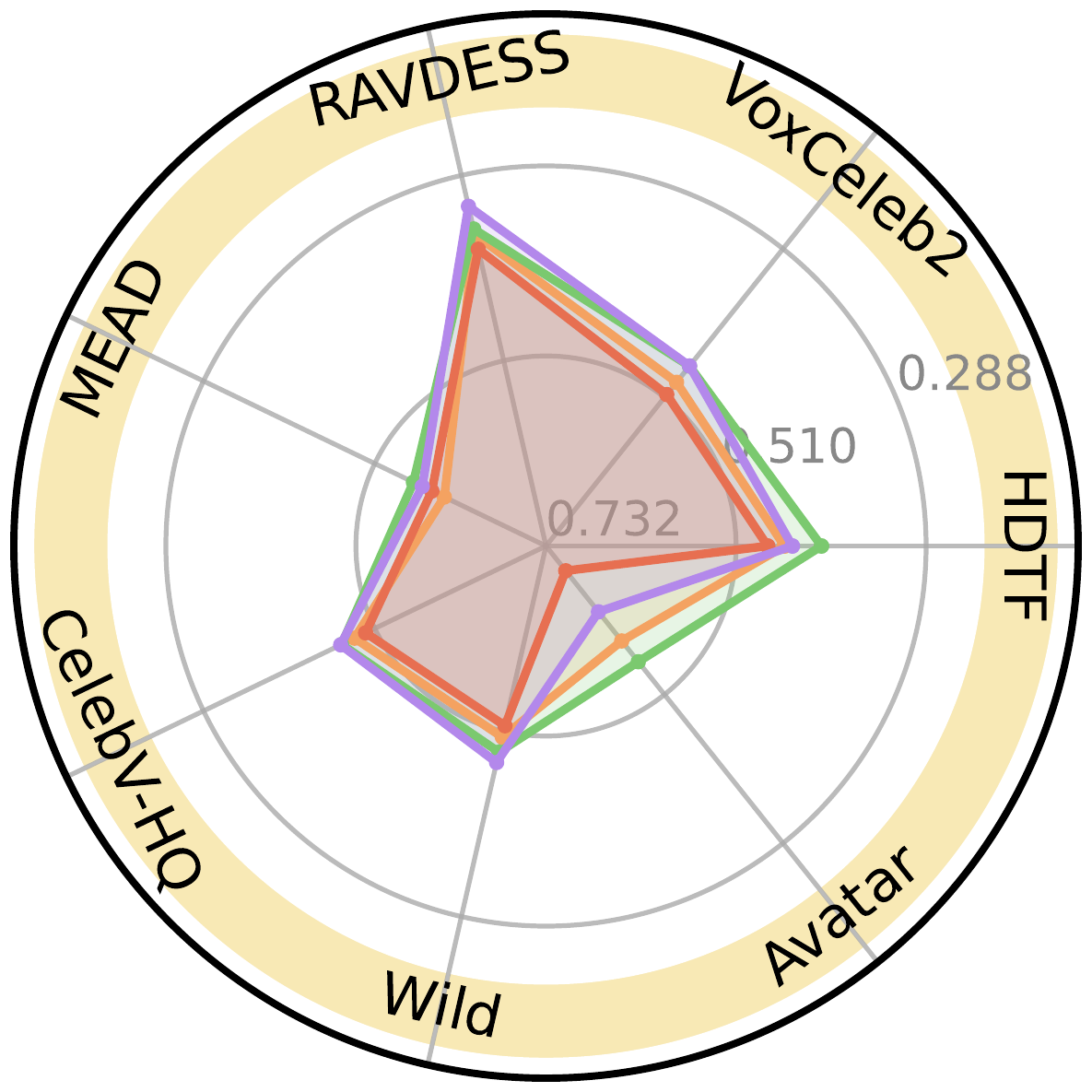}}
\subfloat[Smooth$_{\mathrm{seq}}^{\mathrm{pose}}\downarrow$]{\includegraphics[width=0.16\textwidth]{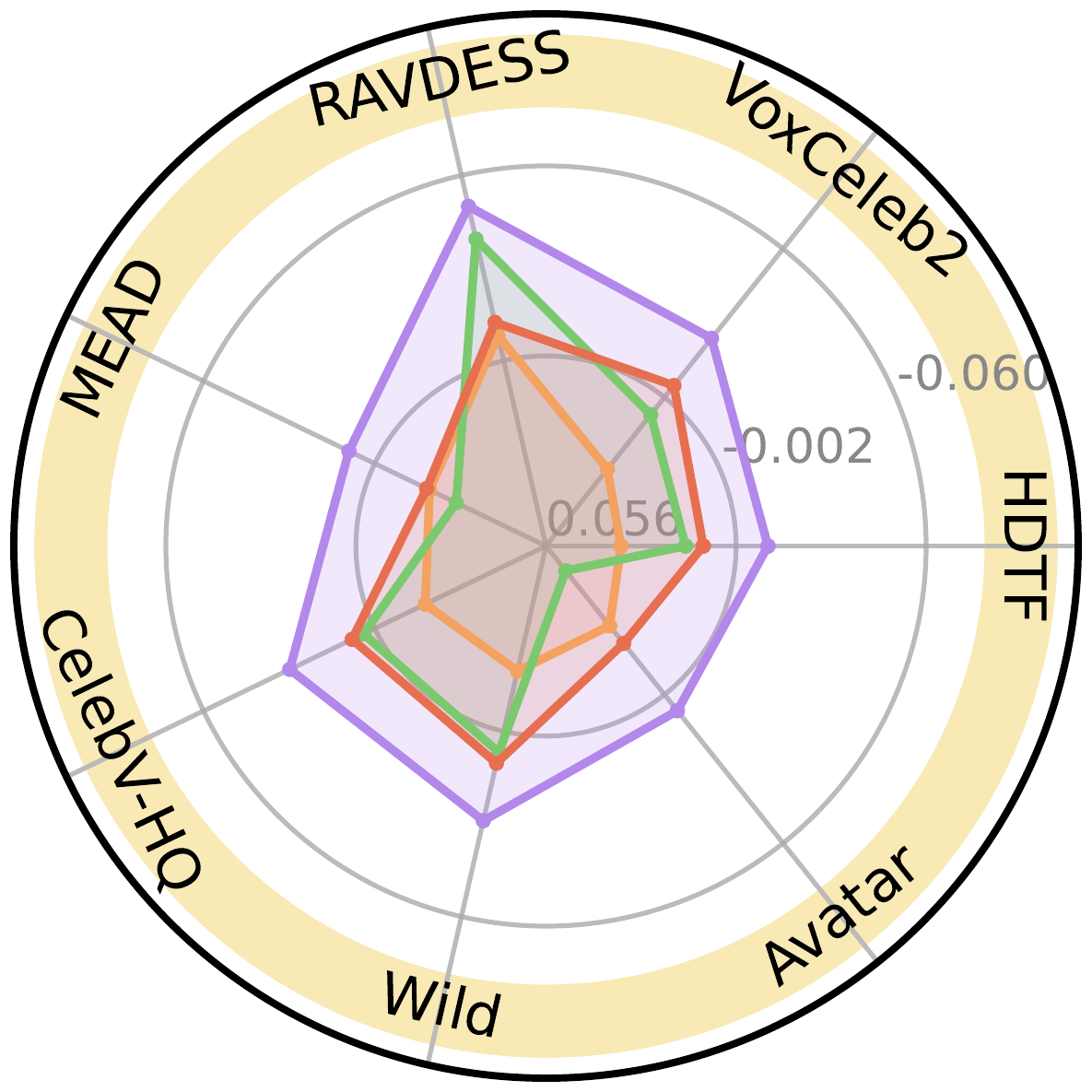}}
\subfloat[Smooth$_{\mathrm{seq}}^{\mathrm{expr}}\downarrow$]{\includegraphics[width=0.16\textwidth]{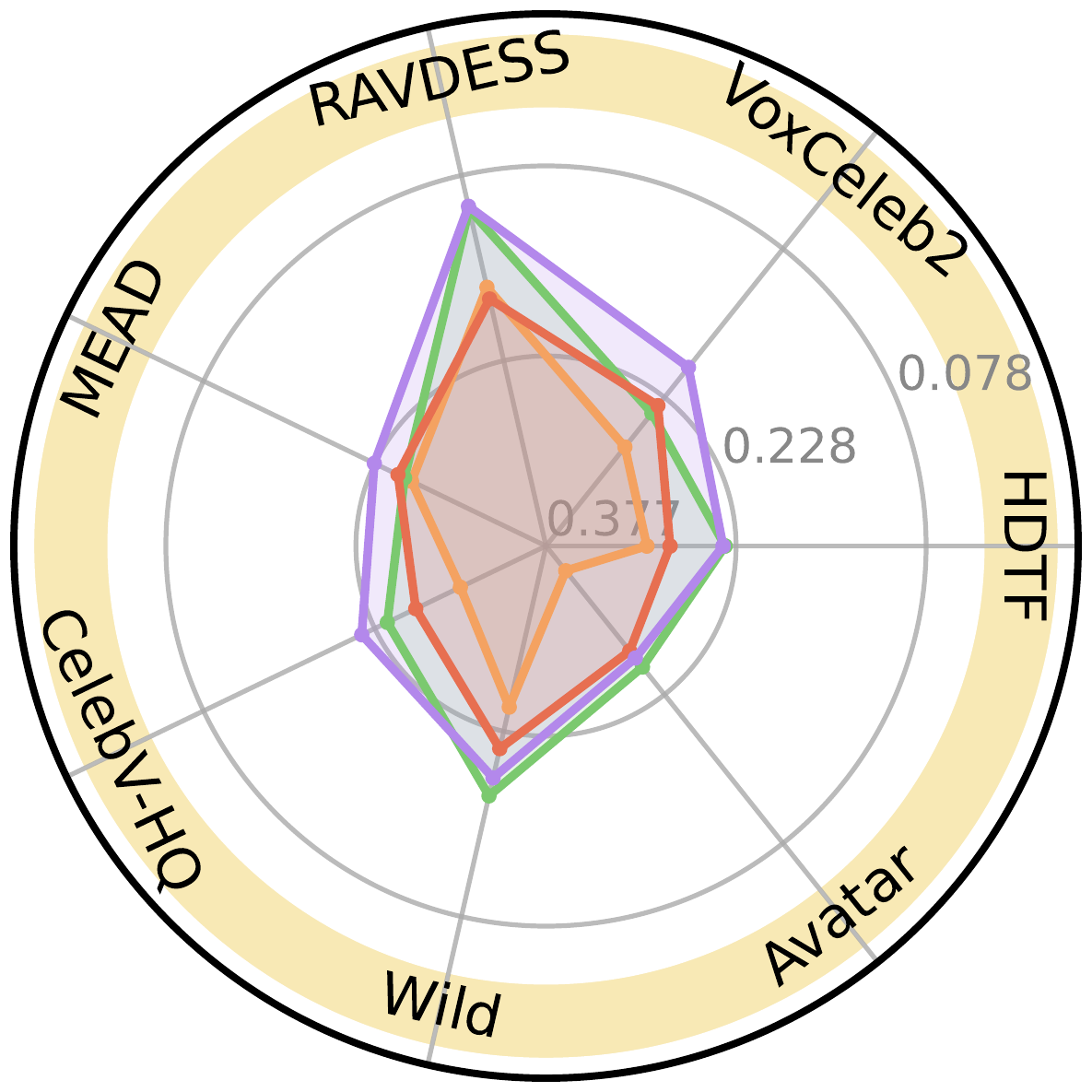}}
    \caption{Dataset-wise radar plots comparing four paradigm categories across 20 metrics. All subfigures share a common legend. Metrics cover distributional fidelity (FID, FVD), perceptual quality (LPIPS, IQA, VQA, CPBD), pixel-level reconstruction (PSNR, SSIM, MS-SSIM, L1), identity preservation (CSIM), audio-visual synchronization (Sync-C, Sync-D), motion naturalness (Smooth), computational efficiency (FPS), and our new temporally-aware evaluation metrics. Lower-is-better metrics are plotted with inverted radial axes so that larger enclosed areas consistently indicate better performance.}
    \label{fig:dataset_radar_16}
\end{figure*}

\subsection{Theoretical Properties of Sequence-Level Alignment}
\label{sec:theory_alignment}

Building on the Soft-DTW formulation above, 
we formalize key theoretical properties of sequence-level alignment. 
Sequence-level similarity is computed via Soft-DTW over monotonic alignment paths $\Pi$ in Eq. \ref{eq:dtw}, 
where each path $\pi$ satisfies $(i,j) \!\to\! (i',j') \!\implies\! i' \ge i, j' \ge j$.

\begin{proposition}[Frame-wise metric as a special case]  
If $\Pi$ is restricted to the identity path $\pi_{\mathrm{diag}} = \{(t,t)\}_{t=1}^{T}$, then
\begin{equation}
    \text{Dist}_{\text{seq}}(\mX, \mY) = \frac{1}{T}\sum_{t=1}^{T} d(\mF_t, \mG_t) = \text{Dist}_{\text{frame}}(\mX, \mY).
\end{equation}
\end{proposition}

\begin{proof}
    With only the identity path, the soft-min reduces to a single sum along the diagonal. Normalization by $T$ recovers exactly the frame-wise distance.
\end{proof}

Frame-wise evaluation corresponds to the special case where each frame is aligned strictly with its counterpart in the reference sequence; Soft-DTW generalizes this by allowing flexible, monotonic alignments across the entire sequence.

\begin{proposition}[Monotonicity and temporal order preservation]
    All paths $\pi \in \Pi$ satisfy $(i',j') \ge (i,j)$. 
\end{proposition}  
 
\begin{proof}
    By definition, admissible Soft-DTW paths cannot violate monotonicity.
\end{proof}

Sequence-level alignment preserves the causal order of frames while accommodating minor temporal shifts.

\begin{proposition}[Bounded sensitivity to temporal shifts]
    Suppose $\mG$ is a temporally shifted version of $\mF$ with shift $\delta$, $|\delta| \le \Delta$, \ie, $\mG_s = \mF_{s+\delta}$ (with boundary handling as needed). Then the optimal alignment path deviates from the identity diagonal by at most $\Delta$, and
\begin{equation}
    \big|\text{Dist}_{\text{seq}}(\mF, \mG) - \text{Dist}_{\text{seq}}(\mF, \mF)\big| = \mathcal{O}(\Delta \cdot \bar{d}),
\end{equation}
where $\bar{d}$ is the typical frame-level distance magnitude.
\label{prop:bound-sensitivity}
\end{proposition}

\begin{proof}
    Each frame in $\mG$ can align to a frame within $\Delta$ indices of its corresponding $\mF$ frame. The Soft-DTW sum along the optimal path thus accumulates a bounded deviation proportional to $\Delta \cdot \bar{d}$. In contrast, frame-wise metrics accumulate misalignment over all $T$ frames.
\end{proof} 


Soft-DTW tolerates small temporal variations without losing sensitivity to meaningful structural changes.


These propositions show that Soft-DTW generalizes frame-wise metrics by preserving temporal order, allowing flexible alignment, and limiting sensitivity to local shifts, providing a principled foundation for temporally-aware evaluation of audio-driven talking-head generation. Fig. \ref{fig:ranking_scatter} compares method rankings between frame- and sequence-level metrics. Many points lie off the diagonal, showing that temporal evaluation uncovers strengths, weaknesses, and paradigm-specific shifts overlooked by conventional frame-level metrics, highlighting the value of our Soft-DTW-based sequence-aware framework. Below, we present ablation studies and full results. 

\section{Results and Analysis}\label{sec4}

\subsection{Sensitivity Analysis to the Soft-DTW Temperature}

The Soft-DTW temperature $\gamma$ controls the elasticity of temporal alignment, interpolating between hard optimal alignment and smooth aggregation over multiple alignment paths (Eq.~\ref{eq:dtw}). As $\gamma \to 0$, Soft-DTW converges to classical DTW, recovering rigid optimal alignment, whereas larger $\gamma$ values produce increasingly smooth, trajectory-level matching.
We evaluate $\gamma$ on CSIM, LPIPS, Sync-D, Smooth$_{\text{pose}}^{\text{seq}}$, and Smooth$_{\text{expr}}^{\text{seq}}$ across both \textit{Wild} and \textit{Avatar} datasets, and compare against the original frame-wise baseline (Fig.~\ref{fig:gamma_all}). For fair comparison across metrics, all values are normalized by their mean.

\textbf{Robustness across $\gamma$.} As shown in the grouped violin plots, sequence-level evaluation remains highly stable over a broad range of $\gamma$, particularly for $\gamma \in [10^{-4}, 0.05]$. The metric distributions largely overlap within this regime, indicating that the temporally-aware formulation is not sensitive to precise temperature selection. This demonstrates that the alignment mechanism is well-conditioned. 

\textbf{Frame-wise \vs sequence-level.} Across all metrics, the rigid frame-wise baseline consistently deviates from the aligned variants, with larger discrepancies observed on \textit{Avatar}. This gap suggests that frame-wise evaluation is more sensitive to temporal misalignment, especially under distribution shifts where phase offsets and motion variability are more pronounced.
The stabilization effect is most evident for synchronization and motion-related metrics. Soft-DTW alignment mitigates sensitivity to local phase offsets while preserving temporal order, consistent with the bounded-shift property in Proposition~\ref{prop:bound-sensitivity}. In contrast, appearance-based metrics such as LPIPS exhibit smaller variation across $\gamma$, reflecting their weaker dependence on temporal alignment. Notably, the proposed pose- and expression-based trajectory metrics remain nearly invariant to $\gamma$, indicating that disentangled motion alignment is structurally robust.
%
We now systematically evaluate 20 methods across seven benchmarks. For each dataset, methods are grouped by modeling paradigm (see Table \ref{tab:talkinghead_methods}, last column), and each metric, including the proposed temporally-aware ones, is averaged across all methods within the same paradigm. These per-dataset, per-paradigm averages are used to generate the radar plots in Fig.~\ref{fig:dataset_radar_16}.
For clarity, lower-is-better metrics are plotted with inverted radial axes, so larger enclosed areas consistently indicate better performance. Fig. \ref{fig:normalized_matrix_heatmap} shows a unified comparison of methods under normalized evaluation metrics.

\subsection{Paradigm-Level Performance}


Paradigm-level performance is highly stable and closely aligned with design intent. No single paradigm dominates across all dimensions; instead, each exhibits consistent, specialized strengths. Lip-centric methods reliably excel in synchronization and pixel-level reconstruction, multi-condition fusion approaches lead in identity preservation, motion-space disentanglement improves short-term smoothness and efficiency, and holistic full-motion pipelines deliver superior perceptual realism and structured temporal coherence. This specialization-driven behavior persists across all canonical datasets, underscoring the importance of paradigm-aware evaluation rather than relying on isolated metrics or single-method comparisons. In the following, we present a detailed analysis.

\textbf{Perceptual quality.} Holistic full-motion and multi-condition fusion paradigms achieve larger radar regions in FID and FVD ($\downarrow$), indicating stronger global realism and temporal coherence. Their advantage stems from modeling long-range facial dynamics and integrating multiple conditioning signals. In contrast, lip-centric and motion-space disentangled approaches exhibit smaller yet stable regions, suggesting that optimizing localized articulation or explicit motion factors does not necessarily improve overall distributional alignment.
A divergence appears between LPIPS and LPIPS$_{\text{seq}}$ ($\downarrow$). At the frame level (LPIPS), lip-centric methods perform better, reflecting sharper and perceptually closer local reconstructions. However, at the sequence level (LPIPS$_{\text{seq}}$), multi-condition fusion becomes competitive or superior, indicating improved temporal perceptual consistency when multiple signals are jointly modeled.
%
IQA and VQA ($\uparrow$) consistently favor holistic full-motion methods, followed by multi-condition fusion. These paradigms produce videos that appear more natural and temporally coherent to human observers. Lip-centric and motion-space disentangled models lag behind, reinforcing that global temporal modeling contributes more to perceived realism than localized optimization.

\textbf{Pixel-level reconstruction.} 
Pixel-based metrics (PSNR, SSIM, MS-SSIM $\uparrow$; L1 $\downarrow$) consistently favor lip-centric methods. This indicates that focusing on the speaking region improves frame-wise reconstruction fidelity relative to ground truth. However, these gains do not translate into superior perceptual or distributional quality, highlighting the known gap between pixel accuracy and human perception.

\textbf{Identity preservation.}  Multi-condition fusion achieves the highest CSIM and CSIM$_{\text{seq}}$ ($\uparrow$) across datasets, demonstrating that explicitly fusing identity-related cues is most effective for preserving speaker identity over time. Lip-centric and holistic full-motion paradigms are comparatively weaker. The similar behavior of CSIM and CSIM$_{\text{seq}}$ suggests that identity consistency is largely stable across temporal scales.

\textbf{Audio-visual synchronization.} Synchronization shows the clearest paradigm separation. Lip-centric methods dominate Sync-C ($\uparrow$) and Sync-D ($\downarrow$), confirming that explicit mouth-focused optimization directly benefits lip-audio alignment. Other paradigms show smaller yet consistent regions, indicating reasonable but less precise alignment. Interestingly, Sync-D$_{\text{seq}}$ shows reduced separation across paradigms, suggesting that sequence-level aggregation mitigates transient frame-level misalignment and narrows apparent performance gaps.

\textbf{Motion naturalness.} For the frame-level Smooth metric ($\uparrow$), motion-space disentangled methods generally perform best, indicating reduced short-term jitter due to explicit motion factorization. On Avatar, lip-centric methods also achieve strong smoothness, likely because synthetic and non-human subjects reduce the need for globally coherent biomechanical modeling, allowing localized articulation control to yield temporally stable trajectories.
In contrast, Smooth$_{\text{seq}}^{\text{pose}}$ and Smooth$_{\text{seq}}^{\text{expr}}$ ($\downarrow$) favor holistic full-motion methods, while motion-space disentanglement performs weaker. This discrepancy suggests that disentanglement improves local stability but may disrupt long-range temporal structure, whereas holistic modeling better preserves coherent pose and expression trajectories over time.

\textbf{Computational efficiency.} FPS ($\uparrow$) reveals the expected complexity trade-off: motion-space disentangled methods are the most efficient, followed by lip-centric, holistic full-motion, and finally multi-condition fusion approaches, which incur additional computational overhead from multi-signal integration.

\begin{figure}[!htbp]
    \centering
    \includegraphics[trim=0 0 0 0, clip=true, width=\linewidth]{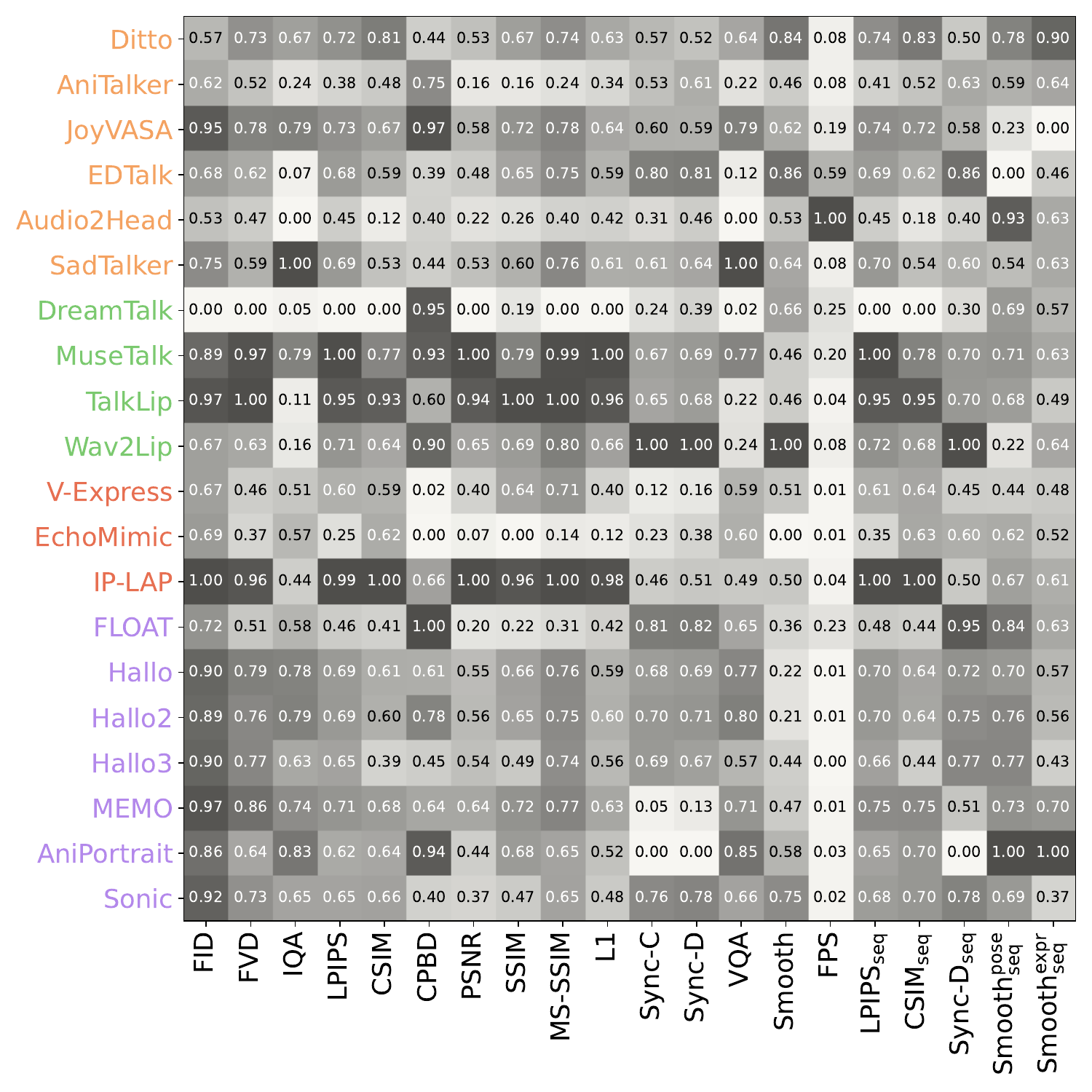}
    \caption{Normalized performance comparison. Each column corresponds to a quantitative metric and each row to a method. Methods are grouped by paradigm using color-coded labels. For each metric, scores are min-max normalized across methods to the range [0, 1], where higher values indicate better performance. Metrics in which lower values denote better quality (\eg, FID, LPIPS) are inverted to ensure consistent interpretation. The figure provides a unified view of cross-metric consistency and paradigm-level performance trends.
    Darker colors indicate better performance.
}
    \label{fig:normalized_matrix_heatmap}
\end{figure}

\subsection{Robustness to Dataset Difficulty and Diversity}

We next analyze paradigm behavior under distribution shift by comparing canonical benchmarks with Wild and Avatar. 

\textbf{Distribution shift.} On \textit{Wild}, which introduces real-world variability such as occlusion, motion blur, lighting inconsistency, and unconstrained head movement, absolute scores slightly fluctuate across metrics but paradigm hierarchies remain stable. No paradigm experiences disproportionate degradation in synchronization, identity preservation, or perceptual quality. The radar patterns show gradual contraction rather than structural reshaping, indicating that real-world difficulty stresses robustness uniformly without altering the relative strengths established on canonical datasets. This suggests that existing paradigms encode stable inductive biases that generalize under realistic noise and motion variability.

\textit{Avatar} produces a qualitatively different shift. By introducing stylized avatars, animated characters, sketches, sculptures, and non-human faces, it departs from photorealistic human statistics while not necessarily increasing noise. Under this stylistic diversity, perceptual and expression-related metrics (IQA, VQA, CPBD, CSIM) tend to increase across paradigms, while distributional metrics such as FID and FVD show moderate decreases. This divergence reflects a decoupling between perceptual sharpness and distributional alignment: stylized or synthetic faces are often visually clean and expressionally exaggerated, improving perceptual measures, yet their appearance statistics deviate from real-video feature distributions, enlarging Fr\'echet distances.

Importantly, synchronization metrics remain comparatively stable under both Wild and Avatar. The radar plots show minimal structural change along Sync-C and Sync-D axes across datasets, indicating that audio-visual alignment is largely independent of appearance realism or stylistic variation. This stability reinforces that synchronization is governed primarily by articulation modeling rather than by global visual statistics.

\textbf{Difficulty \vs diversity.} Across all seven datasets, paradigm ordering remains consistent. Performance shifts are gradual and axis-specific rather than disruptive. Real-world difficulty primarily affects robustness to noise and motion variance, leading to uniform contraction across metrics. Stylistic diversity, in contrast, selectively amplifies perceptual expressiveness while stressing distributional alignment, without degrading synchronization. These distinct effects demonstrate that dataset difficulty and dataset diversity probe fundamentally different generalization mechanisms.

\textbf{Cross-dataset robustness.} Cross-dataset evaluation shows that robustness is structured rather than uniform. Each paradigm retains its characteristic strengths under both forms of distribution shift, and no catastrophic failure is observed. This consistency indicates that lip-focused optimization, motion disentanglement, multi-condition fusion, and holistic temporal modeling encode durable inductive biases rather than brittle dataset-specific heuristics. Evaluating audio-to-video systems under both difficulty-driven and diversity-driven shifts therefore provides a more faithful assessment of generalization than relying on a single notion of harder data.

\subsection{Metric-Level Trade-Off Analysis and Interactions}

The radar patterns indicate that evaluation metrics do not behave independently; instead, they form structured trade-offs that reflect fundamentally different notions of quality. Interpreting results therefore requires understanding how these metrics interact rather than considering them in isolation.

\textbf{Pixel \vs perceptual.} A primary divergence arises between pixel-level reconstruction metrics (PSNR, SSIM, MS-SSIM, L1) and perceptual or distributional metrics (FID, FVD, IQA, VQA). Lip-centric methods consistently achieve superior reconstruction fidelity, yet they do not dominate perceptual realism. Conversely, holistic full-motion pipelines improve FID/FVD and human-perceived quality without maximizing pixel alignment. This confirms that reconstruction metrics reward precise frame-wise similarity, whereas perceptual and distributional metrics evaluate plausibility, coherence, and realism at a higher semantic level. High pixel accuracy does not guarantee perceptual naturalness, and vice versa.

\textbf{Frame \vs sequence.} A second interaction emerges between frame-level and sequence-level metrics. Pairs such as LPIPS \vs LPIPS$_{\text{seq}}$, Sync-D \vs Sync-D$_{\text{seq}}$, and Smooth \vs Smooth$_{\text{seq}}^{\text{pose}}$ / Smooth$_{\text{seq}}^{\text{expr}}$ reveal systematic scale effects. Frame-level metrics are sensitive to instantaneous artifacts, including blur, lip misalignment, or short-term jitter. Sequence-level counterparts aggregate temporal information and emphasize trajectory consistency and long-range coherence. The reduced paradigm separation observed in Sync-D$_{\text{seq}}$ and LPIPS$_{\text{seq}}$ suggests that temporal aggregation mitigates transient deviations, but it may also obscure short-lived instability. These findings demonstrate that frame-level precision and sequence-level structure capture complementary aspects of motion quality.

\textbf{Sync$\neq$realism.} Synchronization metrics further illustrate metric decoupling. Lip-centric approaches dominate Sync-C and Sync-D, confirming that explicit articulation modeling enhances audio-visual alignment. However, this advantage does not automatically translate to stronger perceptual realism or global motion naturalness. Accurate lip-audio correspondence can coexist with limited global coherence, indicating that synchronization alone is not a sufficient proxy for holistic video quality.
Identity preservation metrics (CSIM/CSIM$_{\text{seq}}$) show yet another independent axis. Multi-condition fusion consistently achieves stronger identity consistency, while motion-space disentanglement primarily improves short-term smoothness and efficiency. The limited overlap between these strengths indicates that identity stability, motion regularity, and perceptual realism are governed by different architectural mechanisms and cannot be inferred from one another.

\textbf{Quality \vs speed.} 
Higher perceptual realism and identity preservation are typically achieved by paradigms that incorporate global temporal modeling or multi-signal fusion, at the cost of reduced FPS. Conversely, lighter disentangled models achieve higher efficiency while sacrificing holistic coherence. This quality-efficiency balance highlights that evaluation should remain application-aware.

The metric landscape reflects multiple orthogonal quality dimensions: reconstruction fidelity, perceptual realism, synchronization accuracy, identity stability, temporal coherence, and computational efficiency. No single metric subsumes the others, and improvements along one axis frequently coincide with compromises along another. Robust evaluation of talking head generation therefore requires a multi-dimensional perspective that explicitly acknowledges these interactions and trade-offs rather than relying on isolated indicators.

\subsection{Guidelines, Limitations, and Future Benchmarks}

\textbf{Future method design.} Beyond comparative ranking, our benchmark reveals structural insights into how audio-driven talking-head systems should be evaluated and designed. A central finding is the consistent divergence between frame-level and sequence-level metrics. Several methods that achieve strong per-frame reconstruction or synchronization scores exhibit noticeable instability when assessed with temporally-aware measures. By explicitly quantifying long-range coherence and motion smoothness, the proposed sequence-level metrics expose trajectory inconsistencies that remain invisible under conventional frame-wise evaluation. This observation suggests a necessary shift in evaluation priority: temporal coherence should be treated as a primary objective rather than an emergent byproduct of frame optimization. Future architectures are therefore encouraged to incorporate explicit temporal modeling mechanisms such as memory-based designs, trajectory-level regularization, or coherence-aware loss formulations, and to optimize directly for sequence stability.

Our paradigm-level analysis further clarifies complementary strengths across modeling strategies. Lip-centric approaches consistently achieve superior synchronization accuracy, confirming the effectiveness of localized articulation modeling. However, strong lip alignment alone does not guarantee globally coherent facial dynamics or perceptual realism. In contrast, holistic full-motion pipelines demonstrate stronger perceptual quality and smoother dynamics, highlighting the importance of modeling global spatiotemporal structure. These findings indicate that hybrid designs integrating localized articulation expertise with global temporal modeling may provide a principled path toward balanced performance.
The behavior of multi-condition fusion methods across canonical datasets and our curated \textit{Wild} and \textit{Avatar} subsets reveals another important principle: structured conditioning improves robustness. Explicitly integrating complementary signals, \eg, audio, identity, pose, and geometry, enhances identity preservation and mitigates performance degradation under distribution shifts. When viewed together with temporally-aware evaluation, this suggests that next-generation systems should be both temporally coherent and distributionally resilient.
Importantly, our 15-metric analysis shows that reconstruction fidelity, perceptual realism, synchronization accuracy, motion smoothness, identity preservation, and computational efficiency form partially independent evaluation axes. Improvements along one dimension frequently coincide with degradation along others. The introduction of temporally-aware metrics makes these trade-offs more transparent, discouraging over-optimization of narrow frame-level objectives. Future development should therefore adopt multi-objective or evaluation-aware optimization strategies that explicitly balance competing criteria.

\textbf{Limitations and future benchmarks.} While our benchmark establishes a unified and temporally-aware evaluation protocol, several limitations remain. First, all methods are evaluated using publicly available pretrained models without additional fine-tuning, ensuring fairness and reproducibility but not reflecting adaptation capacity. Incorporating standardized fine-tuning protocols could enable assessment of both generalization and adaptability.
Second, although we evaluate five canonical datasets alongside two curated subsets capturing real-world difficulty and stylistic diversity, additional stress factors such as extreme head poses, heavy occlusions, conversational dynamics, multi-speaker interaction, and cross-lingual audio, remain underexplored. Extending benchmarks along these axes would further disentangle robustness to structural complexity versus distributional variation.
Third, despite using 15 complementary metrics and introducing sequence-level temporally-aware measures, automated evaluation cannot fully capture subjective perceptual quality or communicative expressiveness. Large-scale human preference studies or downstream task-based evaluations would provide valuable complementary perspectives.
Finally, our taxonomy reflects current 2D paradigm families. As diffusion-based, 3D-aware, and multimodal generative approaches mature, the benchmark should evolve accordingly. The proposed framework (metric standardization, temporally-aware sequence evaluation, paradigm grouping, and cross-dataset robustness analysis) provides a scalable foundation for integrating emerging architectures.
By formalizing paradigm-level comparison and elevating temporal coherence to a first-class evaluation principle, this benchmark moves beyond fragmented reporting toward principled, multi-dimensional assessment. We hope it serves not only as a fair comparative resource, but also as a design compass for developing temporally coherent, identity-consistent, and robust talking-head generation systems.

\section{Conclusion}\label{sec5}
We present a unified, evaluation-driven benchmark for audio-driven talking-head generation, addressing inconsistent protocols, fragmented datasets, and heterogeneous reporting. We curate 20 representative methods, organize them by paradigm, and evaluate them across five canonical benchmarks plus new \textit{Wild} and \textit{Avatar} subsets for fair, structured comparison. Our framework integrates 15 complementary metrics covering fidelity, perceptual quality, identity, synchronization, motion, and efficiency, and introduces sequence-level metrics capturing long-range coherence and trajectory consistency.

Empirically, we show that lip-centric methods tend to achieve stronger synchronization and pixel-level reconstruction, multi-condition fusion improves identity preservation and balanced robustness, motion-space disentanglement enhances short-term smoothness and efficiency, and holistic full-motion pipelines deliver superior perceptual realism and long-range temporal coherence. These paradigm characteristics remain largely stable under both real-world difficulty and stylistic diversity shifts, highlighting the presence of durable inductive biases rather than dataset-specific advantages.
Beyond ranking methods, our study exposes fundamental interactions among metrics, such as reconstruction fidelity versus perceptual realism, and frame-level precision versus sequence-level structure, demonstrating that no single metric or paradigm fully captures video quality. By consolidating evaluation protocols, introducing temporally-aware measures, and analyzing cross-dataset robustness, this benchmark establishes a principled foundation for future research in audio-to-video generation.

\backmatter


\bmhead{Acknowledgments}
Zhicheng Zhang is supported by the UNSW University International Postgraduate Award (UIPA).
This work is supported in part by the Australian Research Council (ARC) under Industrial Transformation Research Hub Grant IH180100002. 
This work is also supported by the National Computational Merit Allocation Scheme (NCMAS 2026, NCMAS 2025), with computational resources provided by NCI Australia, an NCRIS-enabled capability supported by the Australian Government.

\bibliography{research}

@article{ma2023dreamtalk,
  title={Dreamtalk: When expressive talking head generation meets diffusion probabilistic models},
  author={Ma, Yifeng and Zhang, Shiwei and Wang, Jiayu and Wang, Xiang and Zhang, Yingya and Deng, Zhidong},
  journal={arXiv preprint arXiv:2312.09767},
  volume={2},
  number={3},
  year={2023}
}

@inproceedings{chen2025echomimic,
  title={Echomimic: Lifelike audio-driven portrait animations through editable landmark conditions},
  author={Chen, Zhiyuan and Cao, Jiajiong and Chen, Zhiquan and Li, Yuming and Ma, Chenguang},
  booktitle={Proceedings of the AAAI Conference on Artificial Intelligence},
  volume={39},
  number={3},
  pages={2403--2410},
  year={2025}
}

@inproceedings{tan2024edtalk,
  title={Edtalk: Efficient disentanglement for emotional talking head synthesis},
  author={Tan, Shuai and Ji, Bin and Bi, Mengxiao and Pan, Ye},
  booktitle={European Conference on Computer Vision},
  pages={398--416},
  year={2024},
  organization={Springer}
}

@article{xu2024hallo,
  title={Hallo: Hierarchical audio-driven visual synthesis for portrait image animation},
  author={Xu, Mingwang and Li, Hui and Su, Qingkun and Shang, Hanlin and Zhang, Liwei and Liu, Ce and Wang, Jingdong and Yao, Yao and Zhu, Siyu},
  journal={arXiv preprint arXiv:2406.08801},
  year={2024}
}

@article{cui2024hallo2,
  title={Hallo2: Long-duration and high-resolution audio-driven portrait image animation},
  author={Cui, Jiahao and Li, Hui and Yao, Yao and Zhu, Hao and Shang, Hanlin and Cheng, Kaihui and Zhou, Hang and Zhu, Siyu and Wang, Jingdong},
  journal={arXiv preprint arXiv:2410.07718},
  year={2024}
}

@article{wang2024v,
  title={V-express: Conditional dropout for progressive training of portrait video generation},
  author={Wang, Cong and Tian, Kuan and Zhang, Jun and Guan, Yonghang and Luo, Feng and Shen, Fei and Jiang, Zhiwei and Gu, Qing and Han, Xiao and Yang, Wei},
  journal={arXiv preprint arXiv:2406.02511},
  year={2024}
}

@inproceedings{prajwal2020lip,
  title={A lip sync expert is all you need for speech to lip generation in the wild},
  author={Prajwal, KR and Mukhopadhyay, Rudrabha and Namboodiri, Vinay P and Jawahar, CV},
  booktitle={Proceedings of the 28th ACM international conference on multimedia},
  pages={484--492},
  year={2020}
}

@article{wei2024aniportrait,
  title={Aniportrait: Audio-driven synthesis of photorealistic portrait animation},
  author={Wei, Huawei and Yang, Zejun and Wang, Zhisheng},
  journal={arXiv preprint arXiv:2403.17694},
  year={2024}
}

@inproceedings{liu2024anitalker,
  title={Anitalker: animate vivid and diverse talking faces through identity-decoupled facial motion encoding},
  author={Liu, Tao and Chen, Feilong and Fan, Shuai and Du, Chenpeng and Chen, Qi and Chen, Xie and Yu, Kai},
  booktitle={Proceedings of the 32nd ACM International Conference on Multimedia},
  pages={6696--6705},
  year={2024}
}

@article{wang2021audio2head,
  title={Audio2head: Audio-driven one-shot talking-head generation with natural head motion},
  author={Wang, Suzhen and Li, Lincheng and Ding, Yu and Fan, Changjie and Yu, Xin},
  journal={arXiv preprint arXiv:2107.09293},
  year={2021}
}

@inproceedings{li2025ditto,
  title={Ditto: Motion-space diffusion for controllable realtime talking head synthesis},
  author={Li, Tianqi and Zheng, Ruobing and Yang, Minghui and Chen, Jingdong and Yang, Ming},
  booktitle={Proceedings of the 33rd ACM International Conference on Multimedia},
  pages={9704--9713},
  year={2025}
}

@inproceedings{ki2025float,
  title={Float: Generative motion latent flow matching for audio-driven talking portrait},
  author={Ki, Taekyung and Min, Dongchan and Chae, Gyeongsu},
  booktitle={Proceedings of the IEEE/CVF International Conference on Computer Vision},
  pages={14699--14710},
  year={2025}
}

@inproceedings{cui2025hallo3,
  title={Hallo3: Highly dynamic and realistic portrait image animation with video diffusion transformer},
  author={Cui, Jiahao and Li, Hui and Zhan, Yun and Shang, Hanlin and Cheng, Kaihui and Ma, Yuqi and Mu, Shan and Zhou, Hang and Wang, Jingdong and Zhu, Siyu},
  booktitle={Proceedings of the Computer Vision and Pattern Recognition Conference},
  pages={21086--21095},
  year={2025}
}

@inproceedings{zhong2023identity,
  title={Identity-preserving talking face generation with landmark and appearance priors},
  author={Zhong, Weizhi and Fang, Chaowei and Cai, Yinqi and Wei, Pengxu and Zhao, Gangming and Lin, Liang and Li, Guanbin},
  booktitle={Proceedings of the IEEE/CVF Conference on Computer Vision and Pattern Recognition},
  pages={9729--9738},
  year={2023}
}

@article{cao2024joyvasa,
  title={JoyVASA: portrait and animal image animation with diffusion-based audio-driven facial dynamics and head motion generation},
  author={Cao, Xuyang and Wang, Guoxin and Shi, Sheng and Zhao, Jun and Yao, Yang and Fei, Jintao and Gao, Minyu},
  journal={arXiv preprint arXiv:2411.09209},
  year={2024}
}

@article{zhang2024musetalk,
  title={Musetalk: Real-time high quality lip synchronization with latent space inpainting},
  author={Zhang, Yue and Minhao, LIU and Chen, Zhaokang and Wu, Bin and Zhan, Chao and He, Yingjie and HUANG, JUNXIN and Zhou, Wenjiang and others},
  year={2024}
}

@inproceedings{zhang2023sadtalker,
  title={Sadtalker: Learning realistic 3d motion coefficients for stylized audio-driven single image talking face animation},
  author={Zhang, Wenxuan and Cun, Xiaodong and Wang, Xuan and Zhang, Yong and Shen, Xi and Guo, Yu and Shan, Ying and Wang, Fei},
  booktitle={Proceedings of the IEEE/CVF conference on computer vision and pattern recognition},
  pages={8652--8661},
  year={2023}
}

@inproceedings{ji2025sonic,
  title={Sonic: Shifting focus to global audio perception in portrait animation},
  author={Ji, Xiaozhong and Hu, Xiaobin and Xu, Zhihong and Zhu, Junwei and Lin, Chuming and He, Qingdong and Zhang, Jiangning and Luo, Donghao and Chen, Yi and Lin, Qin and others},
  booktitle={Proceedings of the Computer Vision and Pattern Recognition Conference},
  pages={193--203},
  year={2025}
}

@inproceedings{wang2023seeing,
  title={Seeing what you said: Talking face generation guided by a lip reading expert},
  author={Wang, Jiadong and Qian, Xinyuan and Zhang, Malu and Tan, Robby T and Li, Haizhou},
  booktitle={Proceedings of the IEEE/CVF Conference on Computer Vision and Pattern Recognition},
  pages={14653--14662},
  year={2023}
}

@article{zheng2024memo,
  title={Memo: Memory-guided diffusion for expressive talking video generation},
  author={Zheng, Longtao and Zhang, Yifan and Guo, Hanzhong and Pan, Jiachun and Tan, Zhenxiong and Lu, Jiahao and Tang, Chuanxin and An, Bo and Yan, Shuicheng},
  journal={arXiv preprint arXiv:2412.04448},
  year={2024}
}

@inproceedings{
zhang2026talkinghead,
title={Talking-Head Generation in Practice},
author={Zhicheng Zhang and Lei Wang and Yongsheng Gao and Yu Zhang},
booktitle={The Second International Workshop on Transformative Insights in Multifaceted Evaluation at The Web Conference 2026},
year={2026},
url={https://openreview.net/forum?id=ns3TgZYQTZ}
}

@article{shen2023difftalk,
  title={Difftalk: Crafting diffusion models for generalized talking head synthesis},
  author={Shen, Shuai and Zhao, Wenliang and Meng, Zibin and Li, Wanhua and Zhu, Zheng and Zhou, Jie and Lu, Jiwen},
  journal={arXiv preprint arXiv:2301.03786},
  volume={2},
  number={4},
  pages={5},
  year={2023}
}

@inproceedings{li2025instag,
  title={InsTaG: Learning Personalized 3D Talking Head from Few-Second Video},
  author={Li, Jiahe and Zhang, Jiawei and Bai, Xiao and Zheng, Jin and Zhou, Jun and Gu, Lin},
  booktitle={Proceedings of the Computer Vision and Pattern Recognition Conference},
  pages={10690--10700},
  year={2025}
}

@inproceedings{zhou2019talking,
  title={Talking face generation by adversarially disentangled audio-visual representation},
  author={Zhou, Hang and Liu, Yu and Liu, Ziwei and Luo, Ping and Wang, Xiaogang},
  booktitle={Proceedings of the AAAI conference on artificial intelligence},
  volume={33},
  number={01},
  pages={9299--9306},
  year={2019}
}

@inproceedings{chung2016out,
  title={Out of time: automated lip sync in the wild},
  author={Chung, Joon Son and Zisserman, Andrew},
  booktitle={Asian conference on computer vision},
  pages={251--263},
  year={2016},
  organization={Springer}
}

@article{chung2018voxceleb2,
  title={Voxceleb2: Deep speaker recognition},
  author={Chung, Joon Son and Nagrani, Arsha and Zisserman, Andrew},
  journal={arXiv preprint arXiv:1806.05622},
  year={2018}
}

@inproceedings{zhang2021flow,
  title={Flow-guided one-shot talking face generation with a high-resolution audio-visual dataset},
  author={Zhang, Zhimeng and Li, Lincheng and Ding, Yu and Fan, Changjie},
  booktitle={Proceedings of the IEEE/CVF conference on computer vision and pattern recognition},
  pages={3661--3670},
  year={2021}
}

@inproceedings{wang2020mead,
  title={Mead: A large-scale audio-visual dataset for emotional talking-face generation},
  author={Wang, Kaisiyuan and Wu, Qianyi and Song, Linsen and Yang, Zhuoqian and Wu, Wayne and Qian, Chen and He, Ran and Qiao, Yu and Loy, Chen Change},
  booktitle={European conference on computer vision},
  pages={700--717},
  year={2020},
  organization={Springer}
}

@article{livingstone2018ryerson,
  title={The Ryerson Audio-Visual Database of Emotional Speech and Song (RAVDESS): A dynamic, multimodal set of facial and vocal expressions in North American English},
  author={Livingstone, Steven R and Russo, Frank A},
  journal={PloS one},
  volume={13},
  number={5},
  pages={e0196391},
  year={2018},
  publisher={Public Library of Science San Francisco, CA USA}
}

@inproceedings{zhu2022celebv,
  title={CelebV-HQ: A large-scale video facial attributes dataset},
  author={Zhu, Hao and Wu, Wayne and Zhu, Wentao and Jiang, Liming and Tang, Siwei and Zhang, Li and Liu, Ziwei and Loy, Chen Change},
  booktitle={European conference on computer vision},
  pages={650--667},
  year={2022},
  organization={Springer}
}

@article{hondru2025masked,
  title={Masked image modeling: A survey},
  author={Hondru, Vlad and Croitoru, Florinel Alin and Minaee, Shervin and Ionescu, Radu Tudor and Sebe, Nicu},
  journal={International Journal of Computer Vision},
  volume={133},
  number={10},
  pages={7154--7200},
  year={2025},
  publisher={Springer}
}

@inproceedings{zhang2018unreasonable,
  title={The unreasonable effectiveness of deep features as a perceptual metric},
  author={Zhang, Richard and Isola, Phillip and Efros, Alexei A and Shechtman, Eli and Wang, Oliver},
  booktitle={Proceedings of the IEEE conference on computer vision and pattern recognition},
  pages={586--595},
  year={2018}
}

@inproceedings{huang2024vbench,
  title={Vbench: Comprehensive benchmark suite for video generative models},
  author={Huang, Ziqi and He, Yinan and Yu, Jiashuo and Zhang, Fan and Si, Chenyang and Jiang, Yuming and Zhang, Yuanhan and Wu, Tianxing and Jin, Qingyang and Chanpaisit, Nattapol and others},
  booktitle={Proceedings of the IEEE/CVF Conference on Computer Vision and Pattern Recognition},
  pages={21807--21818},
  year={2024}
}

@inproceedings{cuturi2017soft,
  title={Soft-dtw: a differentiable loss function for time-series},
  author={Cuturi, Marco and Blondel, Mathieu},
  booktitle={International conference on machine learning},
  pages={894--903},
  year={2017},
  organization={PMLR}
}

@inproceedings{zhou2021pose,
  title={Pose-controllable talking face generation by implicitly modularized audio-visual representation},
  author={Zhou, Hang and Sun, Yasheng and Wu, Wayne and Loy, Chen Change and Wang, Xiaogang and Liu, Ziwei},
  booktitle={Proceedings of the IEEE/CVF conference on computer vision and pattern recognition},
  pages={4176--4186},
  year={2021}
}

@article{croitoru2023diffusion,
  title={Diffusion models in vision: A survey},
  author={Croitoru, Florinel-Alin and Hondru, Vlad and Ionescu, Radu Tudor and Shah, Mubarak},
  journal={IEEE transactions on pattern analysis and machine intelligence},
  volume={45},
  number={9},
  pages={10850--10869},
  year={2023},
  publisher={Ieee}
}

@article{xing2024survey,
  title={A survey on video diffusion models},
  author={Xing, Zhen and Feng, Qijun and Chen, Haoran and Dai, Qi and Hu, Han and Xu, Hang and Wu, Zuxuan and Jiang, Yu-Gang},
  journal={ACM Computing Surveys},
  volume={57},
  number={2},
  pages={1--42},
  year={2024},
  publisher={ACM New York, NY}
}

@inproceedings{fan2022faceformer,
  title={Faceformer: Speech-driven 3d facial animation with transformers},
  author={Fan, Yingruo and Lin, Zhaojiang and Saito, Jun and Wang, Wenping and Komura, Taku},
  booktitle={Proceedings of the IEEE/CVF conference on computer vision and pattern recognition},
  pages={18770--18780},
  year={2022}
}

@inproceedings{thies2020neural,
  title={Neural voice puppetry: Audio-driven facial reenactment},
  author={Thies, Justus and Elgharib, Mohamed and Tewari, Ayush and Theobalt, Christian and Nie{\ss}ner, Matthias},
  booktitle={European conference on computer vision},
  pages={716--731},
  year={2020},
  organization={Springer}
}

@article{zhou2020makelttalk,
  title={Makelttalk: speaker-aware talking-head animation},
  author={Zhou, Yang and Han, Xintong and Shechtman, Eli and Echevarria, Jose and Kalogerakis, Evangelos and Li, Dingzeyu},
  journal={ACM Transactions On Graphics (TOG)},
  volume={39},
  number={6},
  pages={1--15},
  year={2020},
  publisher={ACM New York, NY, USA}
}

@article{har2026heygen,
  title={HeyGen’s AI video platform for English language teaching},
  author={Har, Frankie and Javier, Darren Rey C},
  journal={The Asian Journal of Applied Linguistics},
  volume={10},
  number={1},
  pages={1313--1313},
  year={2026}
}

@article{jin2020live,
  title={A live speech-driven avatar-mediated three-party telepresence system: design and evaluation},
  author={Jin, Aobo and Deng, Qixin and Deng, Zhigang},
  journal={PRESENCE: Virtual and Augmented Reality},
  volume={29},
  pages={113--139},
  year={2020},
  publisher={MIT Press One Rogers Street, Cambridge, MA 02142-1209, USA journals-info~…}
}

@article{christoff2023application,
  title={Application of a 3D talking head as part of telecommunication AR, VR, MR system: Systematic review},
  author={Christoff, Nicole and Neshov, Nikolay N and Tonchev, Krasimir and Manolova, Agata},
  journal={Electronics},
  volume={12},
  number={23},
  pages={4788},
  year={2023},
  publisher={MDPI}
}

@inproceedings{zhu2025infp,
  title={INFP: Audio-driven interactive head generation in dyadic conversations},
  author={Zhu, Yongming and Zhang, Longhao and Rong, Zhengkun and Hu, Tianshu and Liang, Shuang and Ge, Zhipeng},
  booktitle={Proceedings of the IEEE/CVF Conference on Computer Vision and Pattern Recognition},
  pages={10667--10677},
  year={2025}
}

@article{yan2024dialoguenerf,
  title={Dialoguenerf: Towards realistic avatar face-to-face conversation video generation},
  author={Yan, Yichao and Zhou, Zanwei and Wang, Zi and Gao, Jingnan and Yang, Xiaokang},
  journal={Visual Intelligence},
  volume={2},
  number={1},
  pages={24},
  year={2024},
  publisher={Springer}
}

@inproceedings{bai2025survey,
  title={A Survey on Audio-Driven Talking Face Generation},
  author={Bai, Xue and He, Xiangzhen and Ma, Mengdi and Wang, Xiang and Jiang, Wenhao and Du, Tingting and Huang, Zijian},
  booktitle={2025 IEEE 6th International Seminar on Artificial Intelligence, Networking and Information Technology (AINIT)},
  pages={1--6},
  year={2025},
  organization={IEEE}
}

@article{zhen2023human,
  title={Human-computer interaction system: A survey of talking-head generation},
  author={Zhen, Rui and Song, Wenchao and He, Qiang and Cao, Juan and Shi, Lei and Luo, Jia},
  journal={Electronics},
  volume={12},
  number={1},
  pages={218},
  year={2023},
  publisher={MDPI}
}

@article{gowda2023pixels,
  title={From pixels to portraits: A comprehensive survey of talking head generation techniques and applications},
  author={Gowda, Shreyank N and Pandey, Dheeraj and Gowda, Shashank Narayana},
  journal={arXiv preprint arXiv:2308.16041},
  year={2023}
}

@inproceedings{snell2017learning,
  title={Learning to generate images with perceptual similarity metrics},
  author={Snell, Jake and Ridgeway, Karl and Liao, Renjie and Roads, Brett D and Mozer, Michael C and Zemel, Richard S},
  booktitle={2017 IEEE international conference on image processing (ICIP)},
  pages={4277--4281},
  year={2017},
  organization={IEEE}
}

@inproceedings{Zhang2018LPIPS,
  title={The Unreasonable Effectiveness of Deep Features as a Perceptual Metric},
  author={Zhang, Richard and Isola, Phillip and Efros, Alexei A and Shechtman, Eli and Wang, Oliver},
  booktitle={CVPR},
  year={2018}
}

@inproceedings{deng2019arcface,
  title={Arcface: Additive angular margin loss for deep face recognition},
  author={Deng, Jiankang and Guo, Jia and Xue, Niannan and Zafeiriou, Stefanos},
  booktitle={Proceedings of the IEEE/CVF conference on computer vision and pattern recognition},
  pages={4690--4699},
  year={2019}
}

@inproceedings{Chung2016SyncNet,
  title={Lip Reading in the Wild},
  author={Chung, Joon Son and Zisserman, Andrew},
  booktitle={BMVC},
  year={2016}
}

@article{petitjean2011global,
  title={A global averaging method for dynamic time warping, with applications to clustering},
  author={Petitjean, Fran{\c{c}}ois and Ketterlin, Alain and Gan{\c{c}}arski, Pierre},
  journal={Pattern recognition},
  volume={44},
  number={3},
  pages={678--693},
  year={2011},
  publisher={Elsevier}
}

@article{guo2024liveportrait,
  title={Liveportrait: Efficient portrait animation with stitching and retargeting control},
  author={Guo, Jianzhu and Zhang, Dingyun and Liu, Xiaoqiang and Zhong, Zhizhou and Zhang, Yuan and Wan, Pengfei and Zhang, Di},
  journal={arXiv preprint arXiv:2407.03168},
  year={2024}
}

\end{document}